\documentclass{lmcs}
\pdfoutput=1

\usepackage{lastpage}
\lmcsdoi{16}{2}{7}
\lmcsheading{}{\pageref{LastPage}}{}{}%
{Mar.~12,~2019}{Jun.~01,~2020}{}

\usepackage[utf8]{inputenc}


\usepackage[american]{babel}
\usepackage[english=american,strict=true,autostyle=true]{csquotes} 

\usepackage{mathtools}
\usepackage{mleftright}
\usepackage{bbold}
\newcommand{\vvmathbb}[1]{\mathbb{#1}}

\mathtoolsset{centercolon,mathic}

\graphicspath{{drawings/}}
\usepackage{etoolbox}
\usepackage{xpatch}
\usepackage{xstring}
\usepackage{letltxmacro}
\usepackage{xparse}
\usepackage{breakurl}
\usepackage{todonotes}
\usepackage{tikz-cd}
\usepackage{booktabs}
\usepackage{tabu}
\usepackage[inline]{enumitem}
\usepackage{nth}

\newenvironment{tightcenter}
 {\parskip=0pt\par\nopagebreak\centering}
 {\par\noindent\ignorespacesafterend}

\usetikzlibrary{positioning}
\usetikzlibrary{bending}
\usetikzlibrary{arrows.meta}
\usetikzlibrary{decorations.pathreplacing}
\pgfdeclarearrow{
  name = txto,
  setup code = {
    \pgfarrowssettipend{1.9285\pgflinewidth}
    \pgfarrowssetbackend{-2.125\pgflinewidth}
    \pgfarrowssetlineend{0pt}
    \pgfarrowssetvisualbackend{0pt}
    \pgfarrowshullpoint{1.9285\pgflinewidth}{0pt}
    \pgfarrowsupperhullpoint{-1.5714\pgflinewidth}{3.6607\pgflinewidth}
    \pgfarrowsupperhullpoint{-2.125\pgflinewidth}{3.0714\pgflinewidth}
  },
  drawing code = {
    \pgfsetdash{}{0pt}%
    \pgfpathmoveto{\pgfpoint{1.9285\pgflinewidth}{0pt}}%
    \pgfpathlineto{\pgfpoint{-1.5714\pgflinewidth}{3.6607\pgflinewidth}}%
    \pgfpathlineto{\pgfpoint{-2.125\pgflinewidth}{3.0714\pgflinewidth}}%
    \pgfpathlineto{\pgfpoint{0pt}{0.5\pgflinewidth}}%
    \pgfpathlineto{\pgfpoint{0pt}{-0.5\pgflinewidth}}%
    \pgfpathlineto{\pgfpoint{-2.125\pgflinewidth}{-3.0714\pgflinewidth}}%
    \pgfpathlineto{\pgfpoint{-1.5714\pgflinewidth}{-3.6607\pgflinewidth}}%
    \pgfpathclose%
    \pgfusepathqfill
  }
}

\tikzset{>=txto}
\newlength\myfontsize
\newlength\mythinwidth
\newlength\mythickwidth
\newlength\mydoubledist
\tikzset{every picture/.style={line width={\mythinwidth}}}
\tikzset{thin/.style={line width={\mythinwidth}}}
\tikzset{thick/.style={line width={\mythickwidth}}}
\tikzset{arrow/.style={->,node font=\small}}
\tikzset{equiv/.style={-,node font=\small,double,%
  double distance={\mydoubledist},%
  shorten >=-2pt,shorten <=-2pt,%
  arrows={Rectangle[color=white,length=2pt,width=4pt]-Rectangle[color=white,length=2pt,width=4pt]}}}
\tikzset{thick red/.style={thick,color=darkred}}
\tikzset{edgelabel/.style = {midway,font=\small}}
\tikzset{slopedlabel/.style = {edgelabel,sloped}}
\definecolor{darkred}{RGB}{170,0,0}
\tikzcdset{arrow style=tikz}

\makeatletter
\protected\def\tikz@nonactivecolon{\ifmmode\mathrel{\mathop\ordinarycolon}\else:\fi} 
\makeatother

\newcommand{\HoTT}{HoTT}

\newcommand{\Lean}{\textsc{Lean}}

\NewDocumentCommand\NewPairedDelimiter{mmmmO{0mu}O{0mu}}{%
  \NewDocumentCommand#2{mmm}{%
    \IfNoValueTF{##2}
      {\IfBooleanTF{##1}
        {\mskip#6\mleft#3\mskip#5##3\mskip#5\mright#4\mskip#6}
        {\mskip#6#3\mskip#5##3\mskip#5#4\mskip#6}}
      {\mskip#6\mathopen{##2#3}\mskip#5##3\mskip#5\mathclose{##2#4}\mskip#6}%
  }
  \NewDocumentCommand#1{som}{#2{##1}{##2}{##3}}
}
\NewDocumentCommand\NewTripleDelimiter{mmmmmO{0mu}O{0mu}O{0mu}}{%
  \NewDocumentCommand#2{mmmm}{%
    \IfNoValueTF{##2}
      {\IfBooleanTF{##1}
        {\mskip#8\mleft#3\mskip#7##3\mskip#6\middle#4\mskip#6##4\mskip#7\mright#5\mskip#8}
        {\mskip#8#3\mskip#7##3\mskip#6#4\mskip#6##4\mskip#7#5\mskip#8}}
      {\mskip#8\mathopen{##2#3}\mskip#7##3\mskip#6{##2#4}\mskip#6##4\mskip#7\mathclose{##2#5}\mskip#8}%
  }
  \NewDocumentCommand#1{somm}{#2{##1}{##2}{##3}{##4}}
}
\NewDocumentCommand\NewMiddleDelimiter{mmmO{0mu}}{%
  \NewDocumentCommand#2{mmmm}{%
    \IfNoValueTF{##2}
      {\IfBooleanTF{##1}
        {\mleft.##3\mskip#4\middle#3\mskip#4##4\mright.}
        {##3\mskip#4#3\mskip#4##4}}
      {\mathopen{##2.}##3\mskip#4{##2#3}\mskip#4##4\mathclose{##2.}}%
  }
  \NewDocumentCommand#1{somm}{#2{##1}{##2}{##3}{##4}}
}
\NewPairedDelimiter{\parens}{\rawparens}{\lparen}{\rparen}
\NewPairedDelimiter{\brackets}{\rawbrackets}{[}{]}
\NewPairedDelimiter{\braces}{\rawbraces}{\{}{\}}
\NewPairedDelimiter{\angles}{\rawangles}{\langle}{\rangle}
\NewPairedDelimiter{\verts}{\rawverts}{\lvert}{\rvert}[1mu]
\NewPairedDelimiter{\Verts}{\rawVerts}{\lVert}{\rVert}[1mu]
\newcommand{\dummy}{\underline{\hspace{.5em}}}
\mathchardef\rawmathhyphen="2D
\newcommand{\mathhyphen}{\mathrm{\rawmathhyphen}}

\NewTripleDelimiter{\presubst}{\rawpresubst}{[}{/}{]}[1mu]
\NewDocumentCommand\subst{somm}{\mskip1mu\rawpresubst{#1}{#2}{#3}{#4}}

\newcommand{\deq}{\mathrel{\vcentcolon\oldequiv}}

\NewDocumentCommand\NewElim{mm}{%
  \NewDocumentCommand#1{msomsom}{%
    \mathrm{#2}_{##1}\rawbrackets{##2}{##3}{##4}\rawparens{##5}{##6}{##7}}}
\newcommand{\hyphennondep}{\mathrm{\mathhyphen nd}}
\NewElim{\elim}{elim}
\NewElim{\elimnondep}{elim\hyphennondep}

\newcommand{\UU}{\mathcal{U}}

\NewDocumentCommand\pair{somm}{\rawangles{#1}{#2}{#3; #4}}
\DeclareMathOperator{\fst}{fst}
\DeclareMathOperator{\snd}{snd}

\NewDocumentCommand\lam{somo}{%
  \IfNoValueTF{#4}
    {\lambda{#3}.}
    {\lambda\rawparens{#1}{#2}{#3{:}#4}.}}
\newcommand{\fcomp}{\circ}

\newcommand{\point}[1]{\mathrm{pt}_{#1}}

\newcommand{\unittype}{\vvmathbb{1}}

\NewDocumentCommand{\transport}{somsomm}{
  \mathrm{transport}\rawbrackets{#1}{#2}{#3}\rawparens{#4}{#5}{#6; #7}}

\AtBeginDocument{
  \LetLtxMacro{\oldequiv}{\equiv}
  \renewcommand{\equiv}{\simeq}}

\newcommand{\NN}{\vvmathbb{N}}

\newcommand{\sphere}[1]{\vvmathbb{S}^{#1}}
\newcommand{\SOne}{\sphere{1}}

\NewDocumentCommand{\trunc}{somm}{\rawVerts{#1}{#2}{#3}_{#4}}

\NewDocumentCommand{\Tproj}{somm}{\rawverts{#1}{#2}{#3}_{#4}}
\NewDocumentCommand{\Ttransport}{msomsomm}{
  \mathrm{transport}_{#1}\rawbrackets{#2}{#3}{#4}\rawparens{#5}{#6}{#7; #8}}

\NewMiddleDelimiter{\setquot}{\rawsetquot}{/\mskip-1mu}[2mu]

\NewDocumentCommand{\Qproj}{som}{\rawbrackets{#1}{#2}{#3}}

\newcommand{\susp}{\mathrm{susp}}

\newcommand{\cofiber}{\mathrm{cofiber}}
\newcommand{\Cbase}{\mathrm{cfbase}}
\newcommand{\Ccodom}{\mathrm{cfcod}}

\let\quotient\setquot

\newcommand{\pto}{\mathrel{{\cdot}{\to}}}

\newcommand{\wedgetype}[2]{#1 \vee #2}

\newcommand{\bigwedgetype}[2]{\bigvee_{#1} #2}

\newcommand{\booltype}{\vvmathbb{2}}

\newcommand{\ZZ}{\vvmathbb{Z}}

\DeclareMathOperator{\image}{im}
\DeclareMathOperator{\coker}{coker}
\newcommand{\zerogroup}{\vvmathbb{0}}

\newcommand{\EilenbergSteenrod}{Eilenberg--Steenrod}

\newcommand{\proj}{\mathrm{proj}}
\newcommand{\boundary}{\partial}
\newcommand{\coboundary}{\delta}
\NewDocumentCommand{\formalsum}{som}{\ZZ\rawbrackets{#1}{#2}{#3}}
\newcommand{\augmentation}{\epsilon}
\newcommand{\coeff}{\mathrm{coeff}}

\NewDocumentCommand{\axiom}{m}{\textbf{#1}}

\title{Cellular Cohomology in Homotopy Type Theory}
\author{Ulrik Buchholtz}
\email{buchholtz@mathematik.tu-darmstadt.de}
\author{Kuen-Bang {Hou (Favonia)}}
\email{favonia@umn.edu}

\begin{document}

\begin{abstract}
  We present a development of cellular cohomology in homotopy type theory.
  Cohomology associates to each space a sequence of abelian groups
  capturing part of its structure,
  and has the advantage over homotopy groups in that
  these abelian groups of many common spaces are easier to compute.
  Cellular cohomology is a special kind of cohomology designed for cell complexes:
  these are built in stages by attaching spheres of progressively higher dimension,
  and cellular cohomology defines the groups out of the combinatorial description
  of how spheres are attached.
  Our main result is that for finite cell complexes,
  a wide class of cohomology theories
  (including the ones defined through Eilenberg-MacLane spaces)
  can be calculated via cellular cohomology.
  This result was formalized in the Agda proof assistant.
\end{abstract}

\maketitle

\section{Introduction}
\label{sec:introduction}

\emph{Homotopy type theory} (\HoTT)~\cite{hott-as:book} is a new area exploring
the potential of type-theoretic presentations of homotopy theory,
an area originating from the study of topological spaces up to continuous transformation.
It facilitates computer checking of proofs,
and in some case results in the discovery of new proofs.
(See for example~\cite{blakers-massey-in-hott}.)
Many proofs have since been mechanized and checked by
proof assistants such as Agda~\cite{hott-in:agda}, Lean~\cite{hott-in:lean} and Coq~\cite{hott-in:coq}.
\HoTT\ sheds new light on how type theory can be understood
and leads to new type theories based on the resulting insights~%
\cite{cchm,chtt,ccctt-at:csl,chtt.part1,chtt.part2,chtt.part3,chtt.part4,cubicaltt-hits}.

One important concept in homotopy theory is homotopy groups,
groups associated with every (pointed) space
that reveal important structures of that space.
However, it is difficult to calculate higher homotopy groups
for many common spaces, for example spheres.
As a result, homology and cohomology theories arise
as an alternative way to study spaces.

The success of homotopy type theory begins with homotopy groups
partially because of their simple type-theoretic definition.
For example, the homotopy groups $\pi_n(\sphere n)$ were
calculated in~\cite{licata-brunerie-pinsn}.
Over time people have been extending the success
in homotopy and cohomology theories as well.
This paper is yet another milestone we have achieved.

\smallskip

The basic idea of \emph{cohomology theory}
is to study \emph{functions from cycles} in a space.
Functions from the cycles at a fixed dimension $n$ form a group,
which is called the $n$th cohomology group.
There are several ways to define a theory of such groups,
some combinatorial and some axiomatic;
amazingly, the classical theory states that
the two approaches are essentially equivalent.
Our goal is to recreate such an equivalence in \HoTT.

In this paper we focus on a particular class of types,
\emph{CW complexes,}
which come with an explicit description of
how the space is built by iterated attachment of disks.
We will then introduce \emph{cellular cohomology theory,}
a combinatorial cohomology theory
specifically defined for CW complexes.
After that,
we will define
\emph{ordinary \EilenbergSteenrod\ cohomology theory,}
an axiomatic framework for cohomology theory.

\smallskip

This journal paper is an expansion of the conference version~\cite{cohomology-in:hott}
with more explanations and diagrams.
A large part of the text and Agda formulation presented in the conference paper
already appeared in Hou~(Favonia)'s PhD thesis~\cite{favonia-as:thesis}.
However, one critical lemma, Lemma~\ref{lma:cochain}, was stated as a conjecture
and the main theorem was only proved after the thesis.
During the study, a critical component, \emph{degree},
was radically changed to facilitate the proving.

\section{Notation}

We assume that readers are familiar with
common type-theoretic expressions
as in \cite{hott-as:book,favonia-as:thesis}.
In particular, we work in Martin-L\"of type theory
augmented with Voevodsky's \emph{univalence axiom}
and the \emph{pushout} type constructor for
homotopy pushouts.
From these ingredients we can define
cofibers,
the propositional and higher dimensional truncations~\cite{prop-trunc,joinconstruction},
set quotients,
and other higher inductive types used in this work.
Our constructions also work in
the proposed cubical type theories~\cite{cchm,chtt,ccctt-at:csl,chtt.part1,chtt.part2,chtt.part3,chtt.part4,cubicaltt-hits}.
The following is a brief review.

\subsection{Basics}

We often use the word ``space'' as a synonym for ``type''
in order to facilitate the intuition of types as spaces.
The identification type between two elements $x, y$ of type $A$ is denoted by $x=y$.

A pointed type $X$ is equipped with a point $\point{X}$,
and a pointed arrow $f$ from a pointed type $X$ to another pointed type $Y$
is an arrow equipped with an identification from $f\parens{\point{X}}$ to $\point{Y}$.
Pointed arrow types are specially marked as $X \pto Y$
for their importance in calculating degrees~(Section~\ref{sec:degree}).

\subsection{Higher Inductive Types}

Throughout this paper we are going to employ
different higher inductive types that are
definable using only pushouts.
The first one is the \emph{suspension} of type $A$
as the following pushout:
\begin{center}
  \begin{tikzpicture}
    \matrix[column sep=10mm, row sep=10mm] {
      \node (A) {$A$}; &
      \node (ur) {$\unittype$};
      \\
      \node (bl) {$\unittype$}; &
      \node (susp) {$\susp\parens{A}$};
      \\
    };
    \draw [arrow] (A) -- (ur);
    \draw [arrow] (A) -- (bl);
    \draw [arrow] (ur) -- (susp);
    \draw [arrow] (bl) -- (susp);
    \path (A) -- (susp) node [very near end] {$\ulcorner$};
  \end{tikzpicture}
\end{center}
If $f$ is an arrow from $A$ to $B$, then
$\susp\parens{f}$ is the arrow from $\susp\parens{A}$ to $\susp\parens{B}$
arising from the functoriality of $\susp$.
The $n$th \emph{sphere} $\sphere{n}$ is defined to be the iterated suspension $\susp^n\parens{\booltype}$.

Another instance is the \emph{homotopy cofiber} of a pointed
map $f : X \pto Y$.
It is defined to be the following pushout:
\begin{center}
  \begin{tikzpicture}
    \matrix[column sep=15mm, row sep=10mm] {
      \node (X) {$X$}; &
      \node (Y) {$Y$};
      \\
      \node (bl) {$\unittype$}; &
      \node (cofiber) {$\cofiber\parens{f}$};
      \\
    };
    \draw [arrow] (X) -- (Y) node [edgelabel, above] {$f$};
    \draw [arrow] (X) -- (bl);
    \draw [arrow] (Y) -- (cofiber) node [edgelabel, right] {$\Ccodom$};
    \draw [arrow] (bl) -- (cofiber) node [edgelabel, below] {$\Cbase$};
    \path (X) -- (cofiber) node [very near end] {$\ulcorner$};
  \end{tikzpicture}
\end{center}
If the map $f$ is understood,
we write $\quotient YX := \cofiber(f)$.

Finally, given a family of pointed types $\braces{X_i}_{i:I}$ indexed by some type $I$,
the \emph{wedge} $\bigwedgetype{i:I}{X_i}$ is defined to be this pushout:
\begin{center}
  \begin{tikzpicture}
    \matrix[column sep=25mm, row sep=10mm] {
      \node (X) {$I$}; &
      \node (Y) {$\sum_{i:I} X_i$};
      \\
      \node (bl) {$\unittype$}; &
      \node (cofiber) {$\bigwedgetype{i:I}{X_i}$};
      \\
    };
    \draw [arrow] (X) -- (Y) node [edgelabel, above] {$i \mapsto \pair{i}{\point{X_i}}$};
    \draw [arrow] (X) -- (bl);
    \draw [arrow] (Y) -- (cofiber);
    \draw [arrow] (bl) -- (cofiber);
    \path (X) -- (cofiber) node [very near end] {$\ulcorner$};
  \end{tikzpicture}
\end{center}

\smallskip

The $n$th truncation of type $A$, written $\trunc{A}{n}$,
enjoys the universal property that
for any $n$-truncated type $B$,
arrows from $A$ to $B$ uniquely factor through $\trunc{A}{n}$.
Truncations are definable with (iterated) pushouts~\cite{prop-trunc,joinconstruction};
so are set quotients.

\subsection{Groups}

The \emph{quotient group} $\quotient GH$ is defined for any group $G$ and any normal subgroup $H$ of $G$.
For any pointed space $X$,
the $n$th \emph{homotopy group} $\pi_n(X)$ is defined to be the $0$-truncation of the $n$th iterated loop space of $X$.
The \emph{free abelian group} of a set $A$ is written as $\formalsum{A}$.\footnote{In the Agda formulation,
  we allowed $A$ to be an arbitrary type, but then took the set quotient to force the carrier of the resulting group to be a set.}

\section{CW complexes}
\label{sec:cw}
A CW complex (also known as a cell complex) is an inductively defined type
built by attaching cells,
starting from points, lines, faces, and so on.
The description consists of $A_n$ as the index set of cells at dimension $n$,
along with functions $\alpha_n$ denoting how cells are attached.
(We refer to \cite[Chapter~0]{hatcher-at} for a
discussion of CW complexes in a classical context.)
With this combinatorial description at hand,
one may define the cellular cohomology groups
as shown in Section~\ref{sec:cellular-cohomology}.

Many common types can be represented as CW complexes:
the unit type, the spheres, the torus,
and even the real projective spaces~\cite{realprojective}.
One can also build the two-cell complexes obtained by attaching
a cell of dimension $m+1$
to a sphere of dimension $n$
with attaching map given by
an element of $\pi_m(\sphere n)$.
A special case is that of the \emph{Moore spaces}
$M(\ZZ/q\ZZ,n)$ given by the element of $\pi_n(\sphere n)$
corresponding to $q \in \ZZ$ under the
isomorphism $\pi_n(\sphere n) \equiv \ZZ$.
It follows from this description
that $M(\ZZ/q\ZZ,n)$ is the cofiber
of the degree $q$ map $\sphere{n} \to \sphere{n}$.
The two-cell complexes are the simplest
ones that are not spheres,
and they form an important family of examples
and counter-examples,
see Section~\ref{sec:degree} below for an instance of this.

Let $X_0$ be the index set $A_0$ itself as the start
and in general $X_n$ be the construction up to dimension $n$.%
\footnote{This definition was adapted from Buchholtz's work in \Lean~\cite{cellular.complexes-in:utt-by:ulrik}.}
A cell of index $a : A_{n+1}$ at dimension $n+1$ is specified by its boundary in $X_n$,
denoted by the function $\alpha_{n+1}\pair{a}{-}$ from $\sphere{n}$ to $X_n$.
The type $X_{n+1}$ is then the result after attaching all cells at dimension $n+1$ to $X_n$.
Formally, the function $\alpha_{n+1}$ is of type $A_{n+1} \times \sphere{n} \to X_n$
describing the boundary of each cell.
Inductively, the type $X_{n+1}$ is defined to be the following pushout:
(Note that every pushout in homotopy type theory is a homotopy pushout.)
\begin{center}
  \begin{tikzpicture}
    \matrix[row sep=10mm,column sep=10mm] {
      \node (spoke) {$A_{n+1} \times \sphere{n}$};
      &
      \node (hub) {$A_{n+1}$};
      \\
      \node (old) {$X_n$};
      &
      \node (new) {$X_{n+1}$}; \\
    };
    \draw [arrow] (spoke) -- (old) node [edgelabel, left] {$\alpha_{n+1}$};
    \draw [arrow] (spoke) -- (hub) node [edgelabel, above] {$\fst$};
    \draw [arrow] (hub) -- (new);
    \draw [arrow] (old) -- (new);
    \path (spoke) -- (new) node [very near end] {$\ulcorner$};
  \end{tikzpicture}
\end{center}

In this paper, we only work with \emph{finite}
CW complexes whose building process stops at some finite dimension
and for which every $A_n$ is a finite set.
(That is, every $A_n$ is the standard $k$-element set for some $k$.)
Pictorially, a (finite) CW complex is the following iterated pushout,
starting from the type $X_0 \deq A_0$ and ending at some dimension.
\begin{center}
  \begin{tikzpicture}
    \matrix[row sep=8mm,column sep=5mm] {
      &
      \node (spoke1) {$A_{n+1} \times \sphere{n}$};
      &
      \node (hub1) {$A_{n+1}$};
      &[-2mm]
      \node (spoke2) {$A_{n+2} \times \sphere{n+1}$};
      &
      \node (hub2) {$A_{n+2}$};
      \\
      \node (Xpre) {\ldots};
      &
      \node (X0) {$X_n$};
      &&
      \node (X1) {$X_{n+1}$};
      &&
      \node (X2) {\ldots};
      \\
    };
    \draw [arrow] (spoke1) -- (X0) node [edgelabel, right] {$\alpha_{n+1}$};
    \draw [arrow] (spoke1) -- (hub1) node [edgelabel, above] {$\fst$};
    \draw [arrow] (hub1) -- (X1);
    \draw [arrow] (X0) -- (X1);
    \draw [arrow] (spoke2) -- (X1) node [edgelabel, right] {$\alpha_{n+2}$};
    \draw [arrow] (spoke2) -- (hub2) node [edgelabel, above] {$\fst$};
    \draw [arrow] (hub2) -- (X2);
    \draw [arrow] (X1) -- (X2);
    \draw [arrow] (Xpre) -- (X0);
  \end{tikzpicture}
\end{center}

A \emph{pointed} CW complex
additionally requires $A_0$ to be pointed
(and hence $X_0$ and all following pushouts).

\section{Cellular Cohomology}
\label{sec:cellular-cohomology}

Cohomology theory concerns \emph{functions from cycles,}
and one of the best ways to introduce it
is through its dual, \emph{homology theory,} which is about the cycles themselves.
With access to an explicit, combinatorial description of a type,
suitable algebraic structures can be defined for cycles.

To begin with,
a one-dimensional cycle in homology theory (and cohomology theory)
is a \emph{linear combination of oriented lines}
that at each any point has equally many incoming and outgoing lines
counted with multiplicity
(the in-degree equals the out-degree).
\begin{center}
\begingroup%
  \makeatletter%
  \providecommand\color[2][]{%
    \errmessage{(Inkscape) Color is used for the text in Inkscape, but the package 'color.sty' is not loaded}%
    \renewcommand\color[2][]{}%
  }%
  \providecommand\transparent[1]{%
    \errmessage{(Inkscape) Transparency is used (non-zero) for the text in Inkscape, but the package 'transparent.sty' is not loaded}%
    \renewcommand\transparent[1]{}%
  }%
  \providecommand\rotatebox[2]{#2}%
  \newcommand*\fsize{\dimexpr\f@size pt\relax}%
  \newcommand*\lineheight[1]{\fontsize{\fsize}{#1\fsize}\selectfont}%
  \ifx\svgwidth\undefined%
    \setlength{\unitlength}{97.92187767bp}%
    \ifx\svgscale\undefined%
      \relax%
    \else%
      \setlength{\unitlength}{\unitlength * \real{\svgscale}}%
    \fi%
  \else%
    \setlength{\unitlength}{\svgwidth}%
  \fi%
  \global\let\svgwidth\undefined%
  \global\let\svgscale\undefined%
  \makeatother%
  \begin{picture}(1,0.73047955)%
    \lineheight{1}%
    \setlength\tabcolsep{0pt}%
    \put(0,0){\includegraphics[width=\unitlength,page=1]{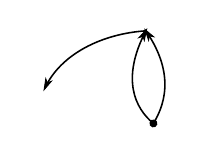}}%
    \put(0.72536939,0.64543216){\color[rgb]{0,0,0}\makebox(0,0)[t]{\smash{\begin{tabular}[t]{c}$x$\end{tabular}}}}%
    \put(0.17089707,0.2822418){\color[rgb]{0,0,0}\makebox(0,0)[rt]{\smash{\begin{tabular}[t]{r}$y$\end{tabular}}}}%
    \put(0,0){\includegraphics[width=\unitlength,page=2]{cycle1.pdf}}%
    \put(0.74999491,0.02364569){\color[rgb]{0,0,0}\makebox(0,0)[t]{\smash{\begin{tabular}[t]{c}$z$\end{tabular}}}}%
    \put(0.40304124,0.52561068){\color[rgb]{0,0,0}\rotatebox{22.116769}{\makebox(0,0)[t]{\smash{\begin{tabular}[t]{c}$a$\end{tabular}}}}}%
    \put(0.37701509,0.03079594){\color[rgb]{0,0,0}\rotatebox{-13.292055}{\makebox(0,0)[t]{\smash{\begin{tabular}[t]{c}$b$\end{tabular}}}}}%
    \put(0.61991573,0.30234712){\color[rgb]{0,0,0}\makebox(0,0)[rt]{\smash{\begin{tabular}[t]{r}$c$\end{tabular}}}}%
    \put(0.82645668,0.32723941){\color[rgb]{0,0,0}\makebox(0,0)[lt]{\smash{\begin{tabular}[t]{l}$d$\end{tabular}}}}%
    \put(0,0){\includegraphics[width=\unitlength,page=3]{cycle1.pdf}}%
  \end{picture}%
\endgroup%

\end{center}
In the above diagram, for example,
$a + b + d$, $c - d$ and $- c - a - b$ are all cycles.
We identify cycles consisting of lines in different orders so that $c - d$ and $- d + c$ are the same cycle.
The intuition is to capture traversals in the space
without worrying about the sequence of lines being traversed;
in fact, lines in a cycle need not be connected.
Each coefficient in a cycle (as a linear combination)
tracks the number of occurrences of a line,
where a negative number signifies a traversal in the opposite direction.
One can define the addition, subtraction and negation on these cycles
as those on linear combinations.

Let $\tilde\boundary_1$ be the function mapping a line from $x$ to $y$ to the linear combination $x - y$,
which represents the oriented boundary of the line.
For example, in the above diagram, we have

\begin{align*}
\tilde\boundary_1\parens{a} &= x - y \\
\tilde\boundary_1\parens{b} &= y - z \\
\tilde\boundary_1\parens{c} &= z - x \\
\tilde\boundary_1\parens{d} &= z - x.
\end{align*}

The criterion of being a cycle, that is,
whether the in- and out- degrees are in balance,
reduces to
whether the summation of $\tilde\boundary_1$ of these lines is exactly zero.
If we extend the function $\tilde\boundary_1$ on lines
to linear combinations of lines as $\boundary_1$,
then \emph{a linear combination of lines is a cycle if and only if it is in the kernel of $\boundary_1$.}
Therefore, we know $a + b + d$ is a cycle because

\begin{equation*}
  \boundary_1\parens{a + b + d} = \boundary_1\parens{a} + \boundary_1\parens{b} + \boundary_1\parens{d}
  = \parens{x - y} + \parens{y - z} + \parens{z - x} = 0.
\end{equation*}

\smallskip

In the context of CW complexes introduced in Section~\ref{sec:cw},
the linear combinations of lines and points are
the free abelian groups $\formalsum{A_1}$ and $\formalsum{A_0}$,
respectively:
\begin{center}
  \begin{tikzpicture}
    \matrix[row sep=10mm,column sep=6mm] {
      &&[-7mm]
      \node (lines) {$\formalsum{A_1}$};
      &
      \node (points) {$\formalsum{A_0}$};
      \\
    };
    \draw [arrow] (lines) -- (points) node [edgelabel, above] {$\boundary_1$};
  \end{tikzpicture}
\end{center}

We would like to identify cycles up to cells at higher dimensions.
In particular, if there is a two-dimensional cell filling
the difference between two cycles,
then those two cycles should be regarded as the same.
Note that for two cycles $u$ and $v$, the difference between them, $u - v$, is still a cycle.
This means we can turn our attention to how a two-dimensional cell fills a cycle.
Intuitively, a cell fills a cycle if its boundary matches the cycle.
For instance, the cell $e$ below fills in the cycle $c - d$:

\begin{center}
\begingroup%
  \makeatletter%
  \providecommand\color[2][]{%
    \errmessage{(Inkscape) Color is used for the text in Inkscape, but the package 'color.sty' is not loaded}%
    \renewcommand\color[2][]{}%
  }%
  \providecommand\transparent[1]{%
    \errmessage{(Inkscape) Transparency is used (non-zero) for the text in Inkscape, but the package 'transparent.sty' is not loaded}%
    \renewcommand\transparent[1]{}%
  }%
  \providecommand\rotatebox[2]{#2}%
  \newcommand*\fsize{\dimexpr\f@size pt\relax}%
  \newcommand*\lineheight[1]{\fontsize{\fsize}{#1\fsize}\selectfont}%
  \ifx\svgwidth\undefined%
    \setlength{\unitlength}{97.92187767bp}%
    \ifx\svgscale\undefined%
      \relax%
    \else%
      \setlength{\unitlength}{\unitlength * \real{\svgscale}}%
    \fi%
  \else%
    \setlength{\unitlength}{\svgwidth}%
  \fi%
  \global\let\svgwidth\undefined%
  \global\let\svgscale\undefined%
  \makeatother%
  \begin{picture}(1,0.73047955)%
    \lineheight{1}%
    \setlength\tabcolsep{0pt}%
    \put(0,0){\includegraphics[width=\unitlength,page=1]{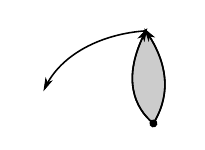}}%
    \put(0.72536939,0.64543216){\color[rgb]{0,0,0}\makebox(0,0)[t]{\smash{\begin{tabular}[t]{c}$x$\end{tabular}}}}%
    \put(0.17089707,0.2822418){\color[rgb]{0,0,0}\makebox(0,0)[rt]{\smash{\begin{tabular}[t]{r}$y$\end{tabular}}}}%
    \put(0,0){\includegraphics[width=\unitlength,page=2]{cycle2.pdf}}%
    \put(0.74999491,0.02364569){\color[rgb]{0,0,0}\makebox(0,0)[t]{\smash{\begin{tabular}[t]{c}$z$\end{tabular}}}}%
    \put(0.40304124,0.52561068){\color[rgb]{0,0,0}\rotatebox{22.116769}{\makebox(0,0)[t]{\smash{\begin{tabular}[t]{c}$a$\end{tabular}}}}}%
    \put(0.37701509,0.03079594){\color[rgb]{0,0,0}\rotatebox{-13.292055}{\makebox(0,0)[t]{\smash{\begin{tabular}[t]{c}$b$\end{tabular}}}}}%
    \put(0.61991573,0.30234712){\color[rgb]{0,0,0}\makebox(0,0)[rt]{\smash{\begin{tabular}[t]{r}$c$\end{tabular}}}}%
    \put(0.82645668,0.32723941){\color[rgb]{0,0,0}\makebox(0,0)[lt]{\smash{\begin{tabular}[t]{l}$d$\end{tabular}}}}%
    \put(0.73015637,0.32087497){\color[rgb]{0,0,0}\makebox(0,0)[t]{\smash{\begin{tabular}[t]{c}$e$\end{tabular}}}}%
    \put(0,0){\includegraphics[width=\unitlength,page=3]{cycle2.pdf}}%
  \end{picture}%
\endgroup%

\end{center}
Similar to $\boundary_1$, we will later define the boundary function $\boundary_2$
for any two-dimensional cell indexed by $w : A_2$
by summing up lines traversed by $\alpha_2\pair{w}{-}$
where $\alpha_2$ is given in the description of the CW complex.
Here, $\boundary_2\parens{e}$ is $c - d$.
A cycle can be filled if and only if
it is in the image of $\boundary_2$.
Consider the following diagram:
\begin{center}
  \begin{tikzpicture}
    \matrix[row sep=10mm,column sep=6mm] {
      \node (2cells) {$\formalsum{A_2}$};
      &
      \node (lines) {$\formalsum{A_1}$};
      &
      \node (points) {$\formalsum{A_0}$};
      \\
    };
    \draw [arrow] (2cells) -- (lines) node [edgelabel, above] {$\boundary_2$} ;
    \draw [arrow] (lines) -- (points) node [edgelabel, above] {$\boundary_1$} ;
  \end{tikzpicture}
\end{center}
The subject of our interest, \emph{cycles up to identifications,}
is exactly the quotient of cycles (the kernel of $\boundary_1$)
by boundaries of cells at the next dimension (the image of $\boundary_2$).
Concretely, in the above example,
the kernel of $\boundary_1$ is generated by the cycles $a + b + c$ and $c - d$,
the image of $\boundary_2$ is generated by the cycle $\boundary_2\parens{e} = c - d$,
and thus the cycles up to identifications constitute
the cyclic group generated by the equivalence class $\brackets{a + b + c}$.
This quotient is called the \emph{first cellular homology group} of type $X$
(with integer coefficients),
and groups for higher dimensions can be defined in a similar way.
How exactly the boundary functions $\boundary_n$ at higher dimensions
should be defined from $A_n$ and $\alpha_n$
will be discussed later.

\phantomsection\label{p:cellular.dual}
The sequence formed by the free abelian groups $\ZZ[A_n]$
and boundary maps $\boundary_n$ is a \emph{chain complex}.
Cellular \emph{co}\/homology takes the dual of the sequence
before calculating the quotients of kernels by images;
it applies the contravariant functor $\hom\parens{-,G}$
for some given abelian group $G$ to the entire sequence.
The dualized sequence is a \emph{cochain complex.}%
\footnote{The kernel-image quotienting is called \emph{homology}
  in the classical literature,
  so the homology of a chain complex is homology,
  and the homology of a cochain complex is cohomology.}
The resulting diagram is
\begin{center}
  \begin{tikzpicture}
    \matrix[column sep=20mm] {
      \node (2cells) {$\hom\parens[\big]{\formalsum{A_2},G}$};
      &
      \node (lines) {$\hom\parens[\big]{\formalsum{A_1},G}$};
      &
      \node (points) {$\hom\parens[\big]{\formalsum{A_0},G}$};
      \\
    };
    \draw [arrow,<-] (2cells) -- (lines) node [edgelabel, above=2mm] {$\hom\parens{\boundary_2,G}$} ;
    \draw [arrow,<-] (lines) -- (points) node [edgelabel, above=2mm] {$\hom\parens{\boundary_1,G}$} ;
  \end{tikzpicture}
\end{center}
and the \emph{first \emph{co}\/homology group}, denoted $H^1(X; G)$, is the quotient
of the kernel of $\hom\parens{\boundary_2,G}$
by the image of $\hom\parens{\boundary_1,G}$.
Groups at higher dimensions are defined in a similar way;
we write $H^n(X; G)$ for the $n$th cellular cohomology group
with coefficients in $G$.

One reasonable definition of $\boundary_{n+1}$ on a cell $b$ at dimension $n+1$
is to individually calculate the coefficient $\coeff(b,a) : \ZZ$ of each cell $a$
within the boundary of the cell $b$;
that is, the boundary function is of the following form
(where the $\sum$ below is the summation in linear algebra, not sum types):
\begin{equation*}
  \boundary_{n+1}(b) \deq \sum_{a : A_n} \coeff(b,a)\;a.
\end{equation*}
If $A_n$ could be infinite
(which is impossible in our current work but possible in future generalization),
in order to make sense of the summation,
it seems we have to assume $\coeff(b,-)$ always has finite support;
this corresponds to the \emph{closure-finiteness} condition in the classical theory,
which is part of the definition of CW complexes and
in fact what the \enquote{C} in the \enquote{CW} stands for.
The classical condition says the boundary of each cell
should be covered by a finite union of cells at lower dimensions,
and so our assumption is well-motivated and may be necessary.

\begin{figure*}
  \centering
  \begin{tikzpicture}
    \matrix[row sep=10mm,column sep=25mm] {
      \node (b) {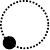};
      &
      \node (Xn) [yshift=1mm] {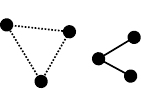};
      &[-15mm]
      \node (Xn/n-1) [yshift=2mm] {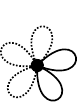};
      &
      \node (a) [yshift=4mm] {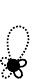};
      \\
    };
    \draw [arrow] (b.east|-b.center) -- (Xn.west|-b.center) node [edgelabel, above] (labelone) {$\alpha_{n+1}\pair{b}{-}$} ;
    \draw [arrow] (Xn.east|-b.center) -- (Xn/n-1.west|-b.center) ;
    \draw [arrow] (Xn/n-1.east|-b.center) -- (a.west|-b.center) node [edgelabel, above] (labelthree) {$\proj_{a}$} ;
  \end{tikzpicture}
  \caption{The function used to define $\coeff(b,a)$.}
  \label{fig:degree-function}
\end{figure*}

Intuitively, the value $\coeff(b,a)$ should capture
the number of (signed) occurrences of $a$
in the boundary of $b$. Considering the CW complex
\begin{center}
\begingroup%
  \makeatletter%
  \providecommand\color[2][]{%
    \errmessage{(Inkscape) Color is used for the text in Inkscape, but the package 'color.sty' is not loaded}%
    \renewcommand\color[2][]{}%
  }%
  \providecommand\transparent[1]{%
    \errmessage{(Inkscape) Transparency is used (non-zero) for the text in Inkscape, but the package 'transparent.sty' is not loaded}%
    \renewcommand\transparent[1]{}%
  }%
  \providecommand\rotatebox[2]{#2}%
  \newcommand*\fsize{\dimexpr\f@size pt\relax}%
  \newcommand*\lineheight[1]{\fontsize{\fsize}{#1\fsize}\selectfont}%
  \ifx\svgwidth\undefined%
    \setlength{\unitlength}{70.05931147bp}%
    \ifx\svgscale\undefined%
      \relax%
    \else%
      \setlength{\unitlength}{\unitlength * \real{\svgscale}}%
    \fi%
  \else%
    \setlength{\unitlength}{\svgwidth}%
  \fi%
  \global\let\svgwidth\undefined%
  \global\let\svgscale\undefined%
  \makeatother%
  \begin{picture}(1,0.63475781)%
    \lineheight{1}%
    \setlength\tabcolsep{0pt}%
    \put(0,0){\includegraphics[width=\unitlength,page=1]{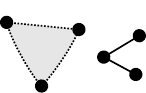}}%
    \put(0.30053833,0.51588711){\color[rgb]{0,0,0}\makebox(0,0)[t]{\lineheight{0}\smash{\begin{tabular}[t]{c}$a$\end{tabular}}}}%
    \put(0.28993,0.2531835){\color[rgb]{0,0,0}\makebox(0,0)[t]{\lineheight{0}\smash{\begin{tabular}[t]{c}$b$\end{tabular}}}}%
  \end{picture}%
\endgroup%

\end{center}
with a two-cell $b$ with boundary consisting of the line $a$ and two other lines,
the value $\coeff(b,a)$ should be $1$ under suitable orientation.
How should we define $\coeff(b,a)$?
The trick is to identify all points and obtain a rose
\begin{center}
\begingroup%
  \makeatletter%
  \providecommand\color[2][]{%
    \errmessage{(Inkscape) Color is used for the text in Inkscape, but the package 'color.sty' is not loaded}%
    \renewcommand\color[2][]{}%
  }%
  \providecommand\transparent[1]{%
    \errmessage{(Inkscape) Transparency is used (non-zero) for the text in Inkscape, but the package 'transparent.sty' is not loaded}%
    \renewcommand\transparent[1]{}%
  }%
  \providecommand\rotatebox[2]{#2}%
  \ifx\svgwidth\undefined%
    \setlength{\unitlength}{36.93745586bp}%
    \ifx\svgscale\undefined%
      \relax%
    \else%
      \setlength{\unitlength}{\unitlength * \real{\svgscale}}%
    \fi%
  \else%
    \setlength{\unitlength}{\svgwidth}%
  \fi%
  \global\let\svgwidth\undefined%
  \global\let\svgscale\undefined%
  \makeatother%
  \begin{picture}(1,1.31889913)%
    \put(0,0){\includegraphics[width=\unitlength,page=1]{boundary2.pdf}}%
    \put(0.60546519,1.07207142){\color[rgb]{0,0,0}\makebox(0,0)[b]{\smash{$a$}}}%
  \end{picture}%
\endgroup%

\end{center}
where the boundary of $b$ is now composed of loops at the center;
the winding number of the loop $a$ is then the coefficient we are looking for.
More precisely, the coefficient is the winding number of the function
depicted in Figure~\ref{fig:degree-function}, or more formally,
\[
  \begin{tikzcd}[column sep=12mm]
    \sphere{1} \ar[r,"\alpha_{2}\pair{b}{-}"] &
    X_1 \ar[r,"\Ccodom_{\phantom{1}}"] &
    \quotient{X_1}{X_{0}} \simeq \bigwedgetype{\dummy:A_1}{\sphere{1}}
    \ar[r,"\proj_{a}"] & \sphere{1}
  \end{tikzcd}
\]
where $\quotient{X_{1}}{X_{0}}$ is the cofiber of the inclusion from $X_{0}$ to $X_1$
and $\proj_{a}$ kills every loop except the one indexed by $a$.
The projection function $\proj_{a}$ is definable whenever
the index type $A_1$ has decidable equality.
For arbitrary dimension $n \geq 1$,
we follow the same pattern to obtain a function from $\sphere{n}$ to $\sphere{n}$,
\[
  \begin{tikzcd}[column sep=15mm]
    \sphere{n}\ar[r,"\alpha_{n+1}\pair{b}{-}"] &
    X_n\ar[r,"\Ccodom_{\phantom{n}}"] &
    \quotient{X_n}{X_{n-1}} \simeq \bigwedgetype{\dummy:A_n}{\sphere{n}}
    \ar[r,"\proj_{a}"]
    & \sphere{n}
  \end{tikzcd}
\]
and then inspect its generalized winding number, namely its \emph{degree}
given by the function $\deg : (\sphere n \to \sphere n) \to \ZZ$
that will be defined in Section~\ref{sec:degree}.
This critically relies on the following lemma which will be used repeatedly in our work:
\begin{lem}
  \label{lma:quot-as-wedge}
  For $n \geq 1$, the types $\quotient{X_n}{X_{n-1}}$ and $\bigwedgetype{\dummy:A_n}{\sphere{n}}$ are equivalent.
\end{lem}
\begin{proof}
  This follows from a direct calculation. Another way to look at it is to recognize the two readings of the following square consisting of two pushout squares. The lower square gives $P \simeq \quotient{X_n}{X_{n-1}}$ and the whole square (which is also a pushout square) gives $P \simeq \bigwedgetype{\dummy:A_n}{\sphere{n}}$, and thus they must be equivalent.
  \begin{center}
    \begin{tikzpicture}
      \matrix[row sep=10mm,column sep=10mm] {
        \node (spoke) {$A_{n} \times \sphere{n-1}$};
        &
        \node (hub) {$A_{n}$};
        \\
        \node (old) {$X_{n-1}$};
        &
        \node (new) {$X_{n}$};
        \\
        \node (one) {$\unittype$};
        &
        \node (quot) {$P$};
        \\
      };
      \draw [arrow] (spoke) -- (old) node [edgelabel, left] {$\alpha_{n}$};
      \draw [arrow] (spoke) -- (hub) node [edgelabel, above] {$\fst$};
      \draw [arrow] (hub) -- (new);
      \draw [arrow] (old) -- (one);
      \draw [arrow] (old) -- (new);
      \draw [arrow] (new) -- (quot);
      \draw [arrow] (one) -- (quot);
      \path (spoke) -- (new) node [very near end] {$\ulcorner$};
      \path (old) -- (quot) node [very near end] {$\ulcorner$};
    \end{tikzpicture}
  \end{center}
  The intuition is that, if we shrink $X_{n-1}$ into a point,
  the remaining of $X_n$ is the attachment of $A_n$ many cells at the $n$th dimension,
  which results into a wedge of $\sphere{n}$ indexed by $A_n$.
\end{proof}
\noindent
The coefficient is then defined as
\begin{equation*}
  \coeff(b,a) \deq \deg\parens[\big]{\proj_{a} \fcomp e\fcomp \Ccodom\fcomp \lam{x} \alpha_{n+1}\pair{b}{x}}
\end{equation*}
where $e$ is the equivalence between $\quotient{X_n}{X_{n-1}}$ and $\bigwedgetype{\dummy:A_n}{\sphere{n}}$.
This finishes our definition of boundary functions in \HoTT.


The cellular cohomology theory that is relevant in our paper
is the \emph{reduced} cellular cohomology theory for pointed types.
A characteristic difference is that a reduced theory
will assign the trivial group instead of the group $\ZZ$ as the zeroth cohomology group of the unit type.
It is arguably more stylish to have trivial groups for the unit type, the most trivial pointed type.
To achieve this,
the \emph{reduced} homology theory augments the chain complex with
$\augmentation$ to $\ZZ$ to give
\[
  \begin{tikzcd}
    \cdots\ar[r] &
    \formalsum{A_1}\ar[r,"\boundary_1"] &
    \formalsum{A_0}\ar[r,"\augmentation"] &
    \ZZ
  \end{tikzcd}
\]
where $\augmentation$ sums up integer coefficients in $\formalsum{A_0}$;
its dual,\phantomsection\label{p:cellular.seq}
\[
  \begin{tikzcd}[column sep=15mm,every label/.append style={inner sep=1.5ex}]
    \cdots &
    \hom\parens[\big]{\formalsum{A_1},G}\ar[l] &[+3mm]
    \hom\parens[\big]{\formalsum{A_0},G}
    \ar[l,"\hom\parens{\boundary_1,G}"'] &[+3mm]
    \hom\parens{\ZZ,G}\ar[l,"\hom\parens{\augmentation,G}"']
  \end{tikzcd}
\]
leads to reduced cohomology groups.
The ending $\ZZ$ effectively kills one degree of freedom in $\formalsum{A_0}$.
In general, the reduced and unreduced integral homology groups
only differ by a $\ZZ$ at the zeroth dimension.

\section{Degrees}
\label{sec:degree}

The notion of the degree of an endo-arrow of spheres
is the key tool for defining and computing with cellular cohomology
(and homology), as shown in the previous section.
We therefore need a definition of degrees in \HoTT.

The requirements of the definition
are that we get for each $n : \NN$
a function $\deg : (\sphere n \to \sphere n) \to \ZZ$
satisfying at least the following properties:
\begin{enumerate}
\item\label{it:deg-id}
  $\deg(\mathrm{id}) = 1$,
  where $\mathrm{id} : \sphere n \to \sphere n$
  is the identity function,
\item\label{it:deg-comp}
  $\deg(g \circ f) = \deg(g)\deg(f)$,
  where $g\circ f$ is the composition of $f$ and $g$,
\item\label{it:deg-susp}
  $\deg(\susp(f)) = \deg(f)$,
  where $\susp(f) : \sphere{n+1} \to \sphere{n+1}$
  is the suspension of $f : \sphere n \to \sphere n$.
\item\label{it:deg-sone}
  $\deg(f)$ is the usual winding number of $f$,
  for $f : \SOne \to \SOne$.
\end{enumerate}
These requirements suffice to uniquely determine the degree function,
although we have not formalized this fact.

One classical definition of degrees relies on homology theories,
which are only now becoming available in \HoTT~\cite{homology-in-hott}.
The cellular homology theory for a CW complex
can be defined by skipping the dualization on page~\ref{p:cellular.dual},
and a more abstract version (in the style described in Section~\ref{sec:eilenberg-steenrod}) can be obtained
by heavily using properties of the smash product,
but in either case it seems to be an unnecessary complication
just for the definition of degrees.

Another approach
is to apply the homotopy group functor $\pi_n$ to the function in question.
Hou~(Favonia)'s thesis~\cite{favonia-as:thesis} employed this.
However, for our purposes where we need good control
on the interaction between the degree and a cohomology theory,
we have found it expedient to use yet another approach.
More precisely, it seems difficult to write down a generator
of the cyclic group $\pi_n\parens{\sphere{n}} \deq \trunc{\Omega_n\parens{\sphere{n}}}{0} \equiv \ZZ$
and thus challenging to relate the group homomorphisms showing up during the proving.

The first observation is that by property~(\ref{it:deg-susp})
we can always reduce the computation of degree
to the case of a pointed map
because $\susp\parens{f}$ is always pointed.
Another manifestation of this is
that the canonical map from the set of pointed endo-arrows of the sphere
$\trunc[\big]{\sphere{n} \pto \sphere{n}}{0}$
to the set of endo-arrows
$\trunc[\big]{\sphere{n} \to \sphere{n}}{0}$
given by forgetting the point
is an equivalence for $n\ge1$.
This means pointedness is free.
For $n=0$, it is not the case:
the identity on $\sphere0$ ($=\booltype$) has degree $1$,
the map that swaps the points has degree $-1$,
while the remaining two maps have degree $0$,
because they suspend to maps of the circle
with winding number zero.

The second idea is to observe
that the degree function is for each $n \geq 1$
a bijection on connected components.
Since $\ZZ$ is a group,
this suggests that we should be able to endow
the type of pointed endo-arrows of $\sphere n$
with a natural group structure
(up to homotopy).
Thus, we want to define a group-like H-space structure
on $\sphere n \pto \sphere n$,
such that the degree map becomes a group isomorphism.
Compared to the previous approach based on $\pi_n$,
where we had difficulty writing down a generator of $\pi_n\parens{\sphere{n}}$,
a generator here is simply the identity function on $\sphere{n}$.

A \emph{co-H-space} is the notion dual to that of an H-space.
This is a natural notion in homotopy theory,
see~\cite{arkowitz} for a survey in the classical setting.
\begin{defi}
  A \emph{co-H-space} is a pointed type $A$
  together with a \emph{comultiplication}
  map $\sigma : A \pto \wedgetype AA$,
  and witnesses for the left and right
  \emph{counit} laws,
  stating that $\pi_1 \circ \sigma$ and
  $\pi_2 \circ \sigma$ are homotopic
  as pointed maps to $\mathrm{id} : A \pto A$,
  where $\pi_i : \wedgetype AA \to A$
  denotes the projection on the left or right summand,
  for $i=1,2$ respectively.
\end{defi}
This definition ensures that the type of pointed maps
from a co-H-space to any pointed type is an H-space.
\begin{defi}
  A \emph{cogroup} is a (homotopy) associative co-H-space
  together with a
  left inverse map $\lambda : A \pto A$
  such that $\nabla \circ (\lambda \vee \mathrm{id}) \circ \sigma$
  is homotopic to the constant pointed map,
  where $\nabla : A \vee A \pto A$ is the folding map,
  which is the identity on each summand.
\end{defi}
Our main interest in cogroups stems from the following fact.
\begin{prop}
  Any suspension $\susp(A)$ has the structure of a
  cogroup with comultiplication
  given by the pinch map
  $\sigma : \susp(A) \to \wedgetype{\susp(A)}{\susp(A)}$.
\end{prop}
The proof is easy and we omit it here.
We note that there are cogroups that are not suspensions,
for example the two-cell complex corresponding to
an element of order $3$ in $\pi_{34}(\sphere 5)$,
see~\cite{berstein-harper}.

The next fact is the main result we need.
\begin{prop}
  If $X$ is a cogroup and $Y$ is a pointed type,
  then the set $\trunc{X \pto Y}0$ has a natural group structure.
\end{prop}
For example, the Moore space $M(\ZZ/q\ZZ,n)$ is a suspension
when $n \ge 2$,
allowing us to define for a pointed type $X$
the \emph{$n$th homotopy group of $X$ with coefficients in $\ZZ/q\ZZ$}
as $\pi_n(X;\ZZ/q\ZZ) \deq \trunc{M(\ZZ/q\ZZ,n) \pto X}{0}$.
(We shall not need these groups for our present purposes.)

Having thus equipped each type $\sphere n \pto \sphere n$ for $n\ge1$
with the structure of a group up to homotopy,
it remains to actually define the degree functions.

For $n=1$ we must choose the winding number per
property~(\ref{it:deg-sone}).
Since we want property~(\ref{it:deg-susp}) to hold,
we use the fact that the action of suspension on pointed maps,
$(X \pto Y) \to (\susp(X) \pto \susp(Y))$
is a group homomorphism whenever $X$ is a cogroup.
So it remains to verify that this map is an group isomorphism
in the case of endo-arrows on $\sphere n$, for $n\ge 1$.
This follows from the Freudenthal suspension theorem~%
\cite[Theorem~8.6.4]{hott-as:book},
in the same way that we prove that $\pi_n(\sphere n) \equiv \ZZ$.
In fact, the loop-suspension adjunction refines to a
group isomorphism $(\susp(X) \pto Y) \equiv (X \pto \Omega Y)$,
using the cogroup structure on $\susp(X)$
and the (H-)group structure on $\Omega Y$.

\section{\EilenbergSteenrod\ Cohomology}
\label{sec:eilenberg-steenrod}

Unlike the above explicit construction,
there is also an \emph{axiomatic} framework for cohomology.
Participants to the special year at Institute for Advanced Study
have brought into \HoTT\ the standard abstract framework
for cohomology theories---\EilenbergSteenrod\ axioms~%
\cite{cohomology-in:utt-by:mike,mayer.vietoris-in:utt,eilenberg.steenrod.cohomology.axioms}.
An (ordinary)\footnote{A cohomology theory is \emph{ordinary}
if it satisfies the last \EilenbergSteenrod\ axiom. See below.}
reduced cohomology theory in \HoTT\ may be defined
as a contravariant functor $h$ from pointed types to sequences of abelian groups
(indexed by integers $\ZZ$) satisfying the axioms below.

Before presenting these axioms, however, we need to define what it means to
\emph{satisfy set-level axiom of choice},
a condition stating that $\prod$ quantifiers and $0$-truncation commute.
This will be used in one of the cohomology axioms shown later.
\begin{defi}[set-level axiom of choice]
  A type $A$ \emph{satisfies the set-level axiom of choice} if,
  for any family of types $B$ indexed by $A$,
  the canonical function from
    $\trunc[\Big]{\prod_{a:A}B(a)}{0}$
  to
    $\prod_{a:A}\trunc[\big]{B(a)}{0}$
  is an equivalence.
\end{defi}
See \cite{hott-as:book} for more discussion about the axiom of choice
and \cite{cohomology-in:utt-by:mike} for its role in cohomology theory in \HoTT.
Essentially, one could present the \EilenbergSteenrod\ axioms \emph{without} the axiom of choice,
but it would be difficult for pointed arrows whose codomains are Eilenberg--Mac~Lane spaces,
an important example of cohomology theories, to satisfy these axioms \emph{within} \HoTT.
In any case,
because we only deal with cell complexes with finite cell sets,
we do not have to worry about the axiom of choice,
for it is provable in \HoTT\ for finite sets.

\smallskip

Anyway, here are the axioms we use.
We write $h^n(X)$ to denote the $n$th group in the sequence for a pointed type $X$.
\begin{description}
  \item[Suspension]
    There is an isomorphism between $h^{n+1}(\susp(X))$ and $h^n(X)$,
    and the choice of isomorphisms is natural in $X$.
  \item[Exactness]
    For any pointed arrow $f : X \pto Y$,
    the following sequence is exact, which means the kernel of $h^n(f)$
    is exactly the image of $h^n(\Ccodom)$.
    \[
      \begin{tikzcd}[column sep=huge]
        h^n(\cofiber(f))\ar[r,"h^n(\Ccodom)"] &
        h^n(Y)\ar[r,"h^n(f)"] &
        h^n(X)
      \end{tikzcd}
    \]
  \item[Wedge]
    Let $I$ be a type satisfying the set-level axiom of choice.
    For any family of pointed types $X$ indexed by $I$,
    the group morphism

    \begin{displaymath}
      \iota^* : h^n\parens[\Bigg]{\bigwedgetype{i:I}{X(i)}} \to \prod_{i:I} h^n(X_i)
    \end{displaymath}

    induced by inclusions $X(i) \to \bigwedgetype{i:I}{X(i)}$ is a group isomorphism.
  \item[Dimension]
    For any integer $n \neq 0$, the group $h^n\parens{\booltype}$ is trivial.
\end{description}

The word \emph{ordinary} refers to satisfying the \axiom{dimension} axiom.
Interesting examples violating this axiom (but satisfying the rest),
such as $K$-theories or complex cobordism,
were discovered after the introduction of the framework,
and are called \emph{extra}\/ordinary cohomology theories.
Our result only handles ordinary ones.

It will be shown that these axioms
uniquely identify cohomology groups for finite CW complexes.
To begin with, we can calculate all the groups of the spheres
directly from the axioms:
\begin{lem}
  \label{lma:sphere}
  For any $m, n : \ZZ$ such that $n \geq 1$, $h^m\parens{\sphere{n}}$
  is isomorphic to $h^0\parens{\booltype}$ if $m = n$ and trivial otherwise.
\end{lem}
\begin{proof}
  Because the spheres are iterated suspensions of $\booltype$,
  one can apply the \axiom{suspension} axiom till it reaches $\booltype$
  and then the \axiom{dimension} axiom if the dimensions mismatch.
\end{proof}
And similarly the groups for the bouquets, or wedges of spheres.
They play an important role in our calculation of the groups of the CW complexes.
\begin{lem}
  \label{lma:bouquet}
  For any $m, n : \ZZ$ such that $n \geq 0$
  and any finite set $A$, $h^m\parens[\big]{\bigwedgetype{\dummy:A}{\sphere{n}}}$
  is isomorphic to $\prod_{\dummy:A} h^0\parens{\booltype}$ if $m = n$ and trivial otherwise.
\end{lem}
There is also an important consequence
from the \axiom{suspension} and the \axiom{exactness} axioms
which will be applied repeatedly:
\begin{lem}
  \label{lma:LES}
  For any pointed arrow $f : X \pto Y$,
  there exist a natural choice of $\gamma_n$ such that the following is a long exact sequence:
  \begin{tightcenter}
    \begin{tikzpicture}
      \matrix[row sep=12mm,column sep=14mm] {
        &&
        \node (4-1) {$\dots$};
        \\
        \node (3-3) {$h^{n+1}(X)$};
        &
        \node (3-2) {$h^{n+1}(Y)$};
        &
        \node (3-1) {$h^{n+1}(\cofiber(f))$};
        \\
        \node (2-3) {$h^n(X)$};
        &
        \node (2-2) {$h^n(Y)$};
        &
        \node (2-1) {$h^n(\cofiber(f))$};
        \\
        \node (1-3) {$\dots$};
        &&
        \\
      };
      \draw [arrow] (1-3) to [out=40,in=220,edge node={node [sloped, above] {$\gamma_{n}$}}] (2-1) ;
      \draw [arrow] (2-1) -- (2-2) node [edgelabel, above] {$h^n(\Ccodom)$} ;
      \draw [arrow] (2-2) -- (2-3) node [edgelabel, above] {$h^n(f)$} ;
      \draw [arrow] (2-3) to [out=40,in=220,edge node={node [slopedlabel, above] {$\gamma_{n+1}$}}] (3-1) ;
      \draw [arrow] (3-1) -- (3-2) node [edgelabel, above] {$h^{n+1}(\Ccodom)$} ;
      \draw [arrow] (3-2) -- (3-3) node [edgelabel, above] {$h^{n+1}(f)$} ;
      \draw [arrow] (3-3) to [out=40,in=220,edge node={node [slopedlabel, above] {$\gamma_{n+2}$}}] (4-1) ;
    \end{tikzpicture}
  \end{tightcenter}
\end{lem}
\begin{proof}
  The key observation is that the cofiber of $\Ccodom : Y \pto \cofiber\parens{f}$
  is equivalent to $\susp\parens{X}$, and thus $h^{n+1}$
  on iterated $\Ccodom$'s is equivalent to $h^{n+1}\parens{\susp(\cdots)}$,
  which by the \axiom{suspension} axiom is isomorphic to $h^n\parens{\cdots}$.
  The exactness of $h^{n+1}$ on iterated $\Ccodom$'s is given by the \axiom{exactness} axiom.
\end{proof}

Finally, let us state the connection between cogroups and cohomology theories:
\begin{prop}
  \label{prop:h-is-hom}
  Let $X$ be a cogroup, and $Y$ a pointed type.
  If $h$ is a cohomology theory, then the map
  \[
    h^n : \trunc[\big]{X \pto Y}0 \to \hom(h^n(Y), h^n(X))
  \]
  is a group homomorphism for each $n : \ZZ$.
\end{prop}
\begin{proof}
  This follows from the commutativity of
  the diagram:
  \[
    \begin{tikzcd}[column sep=small,row sep=small,baseline=(O.base)]
      &[+1em] h^n(Y) \oplus h^n(Y) \ar[r,"h^n(f)\oplus h^n(g)"]
      &[+4em] h^n(X) \oplus h^n(X) \ar[dr,"\pi_1 \cdot \pi_2"] &[+1em] \\
      h^n(Y)\ar[dr,"h^n(\nabla)",swap]\ar[ur,"\Delta"] &&& |[alias=O]| h^n(X) \\
      & h^n(Y \vee Y)\ar[r,"h^n(f\vee g)",swap]\ar[uu,"\equiv"] &
      h^n(X \vee X)\ar[ur,"h^n(\sigma)",swap]\ar[uu,"\equiv"] & \\
    \end{tikzcd}
  \]
  The top row is the addition of $h^n(f)$ and $h^n(g)$ in the group of group homomorphisms from $h^n(Y)$ to $h^n(X)$,
  and the bottom row is $h^n$ applying to the addition of $f$ and $g$ in the group induced by the cogroup structure.
  Commutativity follows from the cogroup structure and the \axiom{wedge} axiom.\footnote{Technically speaking, the binary case of the \axiom{wedge} axiom is derivable from other axioms.}
\end{proof}
This key fact is a generalization of \cite[Lemma~4.60]{hatcher-at}
to cogroups,
and will be used in the proof of Lemma~\ref{lma:cochain}.

\section{Equivalence of Two Cohomology Theories}
\label{sec:equivalence}

Our main result is the following:
\begin{thm}\label{thm:main}
  For any ordinary reduced cohomology theory $h$,
  any pointed finite CW complex $X$ and any $n : \ZZ$,
  $h^n(X)$ is isomorphic to $H^n(X; h^0\parens{\booltype})$.
\end{thm}
This basically states that, on CW complexes, two notions of cohomology coincide.
The significance is that it connects an explicit construction
with a rather abstract framework over a wide range of types.

Note that, when the index $n$ is negative, we can show that $h^n(X)$ is trivial for any CW-complex $X$. (See Section~\ref{sec:reformulate}.) The cellular cohomology groups $H^n(X; G)$ are thus conveniently extended to negative indices $n$ by defining $H^n(X; G)$ to be the trivial group for all $n < 0$. In other words, the theorem is made true for $n < 0$ by construction.

A classical account of the coincidence of these two notions of cohomology may be found in \cite[Chapters~13-15]{may-concise}.
Our approach is to break this theorem into two parts:
\begin{enumerate}
  \item
    We prove that from any ordinary cohomology theory $h$ satisfying the axioms,
    we may reconstruct a cochain complex
    so that its kernel-image quotients
    are the same as the groups directly given by $h$.
  \item
    We show that the reconstructed cochain complex
    is equivalent to the cochain complex used in the cellular cohomology,
    and thus the cellular cohomology groups and $h$ should agree.
\end{enumerate}



\smallskip

In details, these steps are:
\begin{lem}[reformulation of ordinary cohomology groups]
  \label{lma:reformulate}
  For any ordinary reduced cohomology theory $h$
  and any pointed finite CW complex $X$,
  there is a choice of \emph{coboundary functions} $\coboundary_n$
  forming a cochain complex
  \[
    \begin{tikzcd}[column sep=scriptsize]
      \cdots &
      h^n(\quotient{X_n}{X_{n-1}})\ar[l] &
      h^{n-1}(\quotient{X_{n-1}}{X_{n-2}})\ar[l,"\coboundary_n"'] &
      \cdots\ar[l] \\
      & h^1(\quotient{X_1}{X_0})\ar[urr,out=40,in=220] &
      h^0\parens{\booltype} \times h^0(X_0)\ar[l,"\coboundary_1"'] &
      h^0\parens{\booltype}\ar[l,"\coboundary_0"']
    \end{tikzcd}
  \]
  such that $h^n(X)$ is isomorphic to the quotient
  of the kernel of $\coboundary_{n+1}$
  by the image of $\coboundary_n$
  for any $n \geq 0$.
  Moreover, $h^n(X)$ is trivial for $n < 0$.
\end{lem}
The above lemma states that any ordinary cohomology groups
are also the kernel-image quotients of some cochain complex,
similar to cellular cohomology groups.
It is then sufficient to show that the cochain complexes are equivalent.
\begin{lem}[two cochain complexes agree]
  \label{lma:cochain}
  Let $h$ be an ordinary reduced cohomology theory
  and $X$ be a pointed finite CW complex.
  Let $\coboundary_n$ be the group homomorphisms given by
  Lemma~\ref{lma:reformulate} after reformulation.
  There exist an isomorphism
  \[
    k_n : h^n\parens{\quotient{X_{n}}{X_{n-1}}} \equiv \hom\parens[\big]{\formalsum{A_{n}}, h^0\parens{\booltype}}
  \]
  for any $n\geq 1$, an isomorphism for the zeroth dimension
  \[
    k_0 : h^0\parens{\booltype} \times h^0(X_0) \equiv \hom\parens[\big]{\formalsum{A_0}, h^0\parens{\booltype}}
  \]
  and an isomorphism for the augmented part
  \[
    k_{-1} : h^0\parens{\booltype} \equiv \hom\parens[\big]{\ZZ, h^0\parens{\booltype}}
  \]
  such that for any $n \geq 1$, the following square commutes
  \[
    \begin{tikzcd}[column sep=2.5cm,row sep=1.2cm]
      h^{n+1}\parens{\quotient{X_{n+1}}{X_n}}\ar[d,equiv,"k_{n+1}"'] &
      h^{n}\parens{\quotient{X_{n}}{X_{n-1}}}
      \ar[l,"\coboundary_{n+1}"']\ar[d,equiv,"k_n"] \\
      \hom\parens[\big]{\formalsum{A_{n+1}}, h^0\parens{\booltype}} &
      \hom\parens[\big]{\formalsum{A_{n}}, h^0\parens{\booltype}}
      \ar[l,swap,"\hom\parens{\boundary_{n+1}, h^0\parens{\booltype}}"]
    \end{tikzcd}
  \]
  and similarly for the case $n=0$ and the augmented part
  with suitable groups and group homomorphisms.
\end{lem}

Our main result immediately follows from these two lemmas.

\begin{proof}[Proof of Theorem~\ref{thm:main}]
  \newcommand{\newh}{\tilde{h}}
  Let $\newh^n(X)$ be the kernel-image quotients of
  the coboundary maps $\delta$ given by Lemma~\ref{lma:reformulate}.
  We have $\newh^n(X) \simeq h^n(X)$ by the lemma.
  Because the coboundary maps $\delta$ and the dualized boundary maps
  are equivalent according to Lemma~\ref{lma:cochain},
  the resulting groups $\newh^n(X)$ and $H^n\parens{X;h^0\parens{\booltype}}$
  are also isomorphic, and thus the theorem.
\end{proof}

In the following subsections,
we will sketch the proofs of
Lemmas~\ref{lma:reformulate}~and~\ref{lma:cochain}.

\subsection{Reformulation of Cohomology Groups (Lemma~\ref{lma:reformulate})}
\label{sec:reformulate}

The central idea is to construct as many cofibers as possible
from its cellular description,
and then apply the \axiom{exactness} axiom on these cofibers
to obtain long exact sequences by Lemma~\ref{lma:LES}.
From the obtained long exact sequences we can then
calculate the groups of our interest.

\smallskip

Before constructing those cofibers, it is essential to observe that
there is a lemma complimentary to Lemma~\ref{lma:cochain}:
\begin{lem}
  \label{lma:mismatch}
  For any $m \ne n : \ZZ$ such that $n \geq 1$, $h^m\parens{\quotient{X_n}{X_{n-1}}}$ is trivial.
  Moreover, for any $m : \ZZ$ such that $m \ne 0$, $h^m\parens{X_0}$ is also trivial.
\end{lem}
\begin{proof}
  This follows from Lemma~\ref{lma:quot-as-wedge} and the following lemma, and then Lemma~\ref{lma:bouquet}.
\end{proof}
\begin{lem}
  \label{lma:zero-as-wedge}
  Let $a$ be the distinguished point of $X$.
  The types $X_0$ and $\bigwedgetype{\sum_{x:A_0} (a \ne x)}{\booltype}$ are equivalent.
\end{lem}
\begin{proof}
  This is again done by a direct calculation as in the proof of Lemma~\ref{lma:quot-as-wedge}.
  The formulation is different from Lemma~\ref{lma:quot-as-wedge} because
  one of the points---which is the distinguished point $a$ here---%
  is selected as the center of the wedge.
\end{proof}

The way to construct numerous cofiber squares is to consider the following grid diagram
where every grid is a pushout square.
\begin{center}
  \begin{tikzpicture}
    \matrix[row sep=5mm,column sep=3mm] {
      \node (0) {$X_0$}; &
      \node (1) {$X_1$}; &
      \node (2+) {$\dots$}; &
      \node (n) {$X_n$}; &
      \node (n1) {$X_{n+1}$}; &
      \node (n+) {$\dots$};
      \\
      \node (0/0) {$\unittype$}; &
      \node (1/0) {$\quotient{X_1}{X_0}$}; &
      \node (2+/0) {$\dots$}; &
      \node (n/0) {$\quotient{X_n}{X_0}$}; &
      \node (n1/0) {$\quotient{X_{n+1}}{X_0}$}; &
      \node (n+/0) {$\dots$};
      \\
      &
      \node (d/d) {$\ddots$}; &
      \node (2+/d) {$\ddots$}; &
      \node (n/d) {$\vdots$}; &
      \node (n1/d) {$\vdots$}; &
      \node (n+/d) {$\vdots$};
      \\
      &&
      \node (n-1/n-1) {$\unittype$}; &
      \node (n/n-1) {$\quotient{X_{n}}{X_{n-1}}$}; &
      \node (n1/n-1) {$\quotient{X_{n+1}}{X_{n-1}}$}; &
      \node (n+/n-1) {$\dots$};
      \\
      &&&
      \node (n/n) {$\unittype$}; &
      \node (n1/n) {$\quotient{X_{n+1}}{X_n}$}; &
      \node (n+/n) {$\dots$};
      \\
      &&&&
      \node (n1/n1) {$\unittype$}; &
      \node (n+/n1) {$\dots$};
      \\
    };
    \draw [arrow]
      (0) edge (1) edge (0/0)
      (1) edge (2+) edge (1/0)
      (2+) edge (n)
      (n) edge (n1) edge (n/0)
      (n1) edge (n+) edge (n1/0);

    \draw [arrow]
      (0/0) edge (1/0)
      (1/0) edge (2+/0)
      (2+/0) edge (n/0)
      (n/0) edge (n1/0)
      (n1/0) edge (n+/0);

    \draw [arrow]
      (n-1/n-1) edge (n/n-1)
      (n/n-1) edge (n1/n-1) edge (n/n)
      (n1/n-1) edge (n+/n-1) edge (n1/n);

    \draw [arrow]
      (n/n) edge (n1/n)
      (n1/n) edge (n+/n) edge (n1/n1);

    \draw [arrow] (n1/n1) to (n+/n1);
  \end{tikzpicture}
\end{center}
Any square (consisting of one or more grids)
having the unit type $\unittype$ at the bottom left
is a cofiber square
and generates a long exact sequence by Lemma~\ref{lma:LES}.
Three conclusions can be drawn by
choosing different cofiber squares:

First, we can zoom in on a grid on the diagonal:
    \begin{tightcenter}
      \begin{tikzpicture}
        \matrix[row sep=10mm,column sep=5mm] {
          \node (Xn/n-1) {$\quotient{X_{n}}{X_{n-1}}$};
          &
          \node (Xn+1/n-1) {$\quotient{X_{n+1}}{X_{n-1}}$};
          \\
          \node (1) {$\unittype$};
          &
          \node (Xn+1/n) {$\quotient{X_{n+1}}{X_{n}}$.};
          \\
        };
        \draw[arrow]
        (Xn/n-1) edge (1)
        (Xn/n-1) edge (Xn+1/n-1)
        (1) edge (Xn+1/n)
        (Xn+1/n-1) edge (Xn+1/n);
      \end{tikzpicture}
    \end{tightcenter}
    Through Lemma~\ref{lma:LES}, this grid generates the following exact sequence:
    \begin{tightcenter}
      \begin{tikzpicture}
        \matrix[row sep=20mm,column sep=4mm] {
          \node (zero1) {$\zerogroup$};
          &
          \node (ker) {$\ker\parens{\coboundary_{n+1}}$};
          \\[-17mm]
          \node (hnXn+1/n) {$h^n(\quotient{X_{n+1}}{X_n})$};
          &
          \node (hnXn+1/n-1) {$h^n(\quotient{X_{n+1}}{X_{n-1}})$};
          &
          \node (hnXn/n-1) {$h^n(\quotient{X_{n}}{X_{n-1}})$};
          \\[-6mm]
          \node (hn+1Xn+1/n) {$h^{n+1}(\quotient{X_{n+1}}{X_{n}})$};
          &
          \node (hn+1Xn+1/n-1) {$h^{n+1}(\quotient{X_{n+1}}{X_{n-1}})$};
          &
          \node (hn+1Xn/n-1) {$h^{n+1}(\quotient{X_{n}}{X_{n-1}})$.};
          \\[-17mm]
          &
          \node (coker) {$\coker\parens{\coboundary_{n+1}}$};
          &
          \node (zero2) {$\zerogroup$};
          \\
        };
        \draw
        (hnXn+1/n) edge[equiv] (zero1)
        (hnXn+1/n-1) edge[equiv] (ker);
        \draw[arrow]
        (hnXn+1/n) edge (hnXn+1/n-1)
        (hnXn+1/n-1) edge[>->] (hnXn/n-1)
        (hnXn/n-1) to [out=-140,in=40,edge node={node [edgelabel, above] {$\coboundary_{n+1}$}}] (hn+1Xn+1/n)
        (hn+1Xn+1/n) edge[->>] (hn+1Xn+1/n-1)
        (hn+1Xn+1/n-1) edge (hn+1Xn/n-1) ;
        \draw
        (hn+1Xn+1/n-1) edge[equiv] (coker)
        (hn+1Xn/n-1) edge[equiv] (zero2);
      \end{tikzpicture}
    \end{tightcenter}
    We choose the coboundary function $\coboundary_{n+1}$ to be the middle function in the diagram.
    Because $h^n(\quotient{X_{n+1}}{X_n})$ is trivial,
    from the exactness we know $h^n(\quotient{X_{n+1}}{X_{n-1}})$
    is isomorphic to the kernel of $\coboundary_{n+1}$
    and the group homomorphism from it is injective.
    Dually, we know $h^{n+1}(\quotient{X_{n+1}}{X_{n-1}})$
    is isomorphic to the cokernel of $\coboundary_{n+1}$
    and the group homomorphism to it is surjective.

  Secondly,
    let's turn our focus to this square:
    \begin{tightcenter}
      \begin{tikzpicture}
        \matrix[row sep=10mm,column sep=5mm] {
          \node (Xm) {$X_{m}$};
          &
          \node (Xm+1) {$X_{m+1}$};
          \\
          \node (1) {$\unittype$};
          &
          \node (Xm+1/m) {$\quotient{X_{m+1}}{X_{m}}$};
          \\
        };
        \draw[arrow]
        (Xm) edge (1)
        (Xm) edge (Xm+1)
        (1) edge (Xm+1/m)
        (Xm+1) edge (Xm+1/m);
      \end{tikzpicture}
    \end{tightcenter}
    which by Lemma~\ref{lma:LES} gives this exact sequence:
    \[
      \begin{tikzcd}[column sep=3mm]
        h^n(\quotient{X_{m+1}}{X_m})\ar[r] &
        h^n(X_{m+1})\ar[r] &
        h^n(X_{m})\ar[r] &
        h^{n+1}(\quotient{X_{m+1}}{X_m})
      \end{tikzcd}
    \]
    When $n \notin \braces{m,m+1}$, both $h^n(\quotient{X_{m+1}}{X_m})$ and $h^{n+1}(\quotient{X_{m+1}}{X_m})$ are trivial
    by Lemma~\ref{lma:mismatch};
    therefore, by the exactness of the above sequence, $h^n(X_{m+1}) \equiv h^n(X_{m})$.
    The insight is that cells at dimensions much higher or much lower than $n$
    are irrelevant to the cohomology group at dimension $n$.
    This implies that there are at most three different values of $h^n(X_m)$ up to isomorphism:
    \begin{enumerate}
      \item $h^n(X_{n-1}) \equiv h^n(X_{n-2}) \equiv \cdots \equiv h^n(X_{0}) \equiv \zerogroup$.
        The intuition is that $X_m$ for any $m<n$ does not have any interesting information at dimension $n$.
      \item $h^n(X_{n})$.
        $X_n$ has the cells at dimension $n$, but lacks the cells at dimension $\parens{n+1}$
        which may identify some cycles at dimension $n$.
      \item $h^n(X_{n+1}) \equiv h^n(X_{n+2}) \equiv \cdots \equiv h^n(X)$.
        Cells at dimension $\parens{n+2}$ or above play no role in the $n$th cohomology group. This also implies $h^n(X) \equiv h^n(X_0) \equiv \zerogroup$ when $n < 0$.
    \end{enumerate}
    It is thus sufficient to study $h^n(X_{n+1})$ for the $n$th cohomology group of $X$.

  Finally, we investigate the square
    \begin{tightcenter}
      \begin{tikzpicture}
        \matrix[row sep=10mm,column sep=5mm] {
          \node (Xn-2) {$X_{n-2}$};
          &
          \node (Xn+1) {$X_{n+1}$};
          \\
          \node (1) {$\unittype$};
          &
          \node (Xn+1/n-2) {$\quotient{X_{n+1}}{X_{n-2}}$};
          \\
        };
        \draw[arrow]
        (Xn-2) edge (1)
        (Xn-2) edge (Xn+1)
        (1) edge (Xn+1/n-2)
        (Xn+1) edge (Xn+1/n-2);
      \end{tikzpicture}
    \end{tightcenter}
    which generates the exact sequence
    \[
      \begin{tikzcd}[column sep=4mm]
        h^{n-1}(X_{n-2})\ar[r] &
        h^n(\quotient{X_{n+1}}{X_{n-2}})\ar[r] &
        h^n(X_{n+1})\ar[r] &
        h^n(X_{n-2})
      \end{tikzcd}
    \]
    From the previous cofiber square we know both $h^{n-1}(X_{n-2})$ and $h^n(X_{n-2})$ are trivial,
    and again by the exactness $h^n(\quotient{X_{n+1}}{X_{n-2}}) \equiv h^n(X_{n+1})$.
    Therefore, we have
    \[
      h^n(\quotient{X_{n+1}}{X_{n-2}}) \equiv h^n(X_{n+1}) \equiv h^n(X_{n+2}) \equiv \cdots \equiv h^n(X).
    \]
    This means it is sufficient to calculate $h^n(\quotient{X_{n+1}}{X_{n-2}})$.

Combining these three observations, we have the following commutative square for $n\geq 2$:
\begin{center}
  \begin{tikzpicture}
    \matrix[row sep=10mm,column sep=7mm,column 1/.style={anchor=base east},column 2/.style={anchor=base west}] {
      \node (coker) {$\coker\parens{\coboundary_n} \simeq h^n(\quotient{X_{n}}{X_{n-2}})$};
      &
      \node (answer) {$h^n(\quotient{X_{n+1}}{X_{n-2}}) \simeq h^n(X)$};
      \\
      \node (wedge) {$h^n(\quotient{X_{n}}{X_{n-1}})$};
      &
      \node (ker) {$h^n(\quotient{X_{n+1}}{X_{n-1}}) \simeq \ker\parens{\coboundary_{n+1}}$.};
      \\
    };
    \draw[arrow]
      (ker) edge[>->] (wedge)
      (ker) edge[transform canvas={xshift=-8mm}] (ker|-answer.south)
      (wedge) edge[->>,transform canvas={xshift=0mm}] (wedge|-coker.south)
      (answer) edge (coker);
  \end{tikzpicture}
\end{center}
We can further infer that the top homomorphism is injective and the right one is surjective
by applying Lemma~\ref{lma:LES} (and then Lemma~\ref{lma:mismatch}) to the following two cofiber squares, respectively,
\begin{tightcenter}
  \begin{tikzpicture}
    \matrix[row sep=10mm,column sep=3mm] {
      \node (Xn/n-2) {$\quotient{X_{n}}{X_{n-2}}$};
      &
      \node (Xn+1/n-2) {$\quotient{X_{n+1}}{X_{n-2}}$};
      \\
      \node (1) {$\unittype$};
      &
      \node (Xn+1/n) {$\quotient{X_{n+1}}{X_{n}}$};
      \\
    };
    \draw[arrow]
    (Xn/n-2) edge (1)
    (Xn/n-2) edge (Xn+1/n-2)
    (1) edge (Xn+1/n)
    (Xn+1/n-2) edge (Xn+1/n);
  \end{tikzpicture}
  \hspace{.5em}
  \begin{tikzpicture}
    \matrix[row sep=10mm,column sep=3mm] {
      \node (Xn-1/n-2) {$\quotient{X_{n-1}}{X_{n-2}}$};
      &
      \node (Xn+1/n-2) {$\quotient{X_{n+1}}{X_{n-2}}$};
      \\
      \node (1) {$\unittype$};
      &
      \node (Xn+1/n-1) {$\quotient{X_{n+1}}{X_{n-1}}$};
      \\
    };
    \draw[arrow]
    (Xn-1/n-2) edge (1)
    (Xn-1/n-2) edge (Xn+1/n-2)
    (1) edge (Xn+1/n-1)
    (Xn+1/n-2) edge (Xn+1/n-1);
  \end{tikzpicture}
\end{tightcenter}
and finally obtain the desired isomorphism
\[ h^n(X) \simeq \quotient{\ker\parens{\coboundary_{n+1}}}{\image\parens{\coboundary_n}}\]
by the following lemma from group theory:
\begin{lem}
  Let $Q \subseteq P$ be two subgroups of $G$ where $Q$ is normal.
  If we have a group $K$ and a commutative diagram as follows,
  where the group homomorphism from $P$ to $G$ is the canonical inclusion
  and the one from $G$ to $\quotient{G}{Q}$ is the quotienting,
  then $K \simeq \quotient{P}{Q}$.
  \begin{center}
    \begin{tikzpicture}
      \matrix[row sep=10mm,column sep=20mm] {
        \node (coker) {$\quotient{G}{Q}$};
        &
        \node (answer) {$K$};
        \\
        \node (wedge) {$G$};
        &
        \node (ker) {$P$};
        \\
      };
      \draw[arrow]
      (ker) edge[>->] (wedge)
      (ker) edge[->>] (answer)
      (wedge) edge[->>] (coker)
      (answer) edge[>->] (coker);
    \end{tikzpicture}
  \end{center}
\end{lem}
By choosing $P$ to be the kernel and $Q$ the image,
we have the desired formula that $K$ is the quotient.
($Q$ is normal because it is a subgroup of an abelian group.)
This shows that for any $n \geq 2$,
$h^n(X)$ is isomorphic to the kernel-image quotient of adjacent coboundary functions $\coboundary_n$.
The cases for $n=0$ or $1$ can also be derived from the grid diagram similarly
but with special care of the distinguished point of $X_0 \deq A_0$
as demonstrated in the proof of Lemma~\ref{lma:zero-as-wedge}.
This concludes the proof of Lemma~\ref{lma:reformulate}.

\subsection{Equivalence of Two Cochain Complexes (Lemma~\ref{lma:cochain})}

Recall the isomorphisms and commutative squares we want:
\begin{align*}
  k_n : h^n\parens{\quotient{X_{n}}{X_{n-1}}} &\equiv \hom\parens[\big]{\formalsum{A_{n}}, h^0\parens{\booltype}}
  \\
  k_0 : h^0\parens{\booltype} \times h^0(X_0) &\equiv \hom\parens[\big]{\formalsum{A_0}, h^0\parens{\booltype}}
  \\
  k_{-1} : h^0\parens{\booltype} &\equiv \hom\parens[\big]{\ZZ, h^0\parens{\booltype}}
\end{align*}
  \[
    \begin{tikzcd}[column sep=2.5cm,row sep=1.2cm]
      h^n\parens{\quotient{X_{n+1}}{X_n}}\ar[d,equiv,"k_{n+1}"'] &
      h^{n-1}\parens{\quotient{X_{n}}{X_{n-1}}}
      \ar[l,"\coboundary_{n+1}"']\ar[d,equiv,"k_n"] \\
      \hom\parens[\big]{\formalsum{A_{n+1}}, h^0\parens{\booltype}} &
      \hom\parens[\big]{\formalsum{A_{n}}, h^0\parens{\booltype}}
      \ar[l,swap,"\hom\parens{\boundary_{n+1}, h^0\parens{\booltype}}"]
    \end{tikzcd}
  \]

  The isomorphisms $k_n$ for $n \geq 0$ arise from the type equivalences given by Lemmas~\ref{lma:quot-as-wedge}~and~\ref{lma:zero-as-wedge} and then the group isomorphisms by Lemma~\ref{lma:bouquet}.
  The isomorphism $k_{-1}$ is the canonical isomorphism between
  $h^0\parens{\booltype}$ and $\hom\parens{\ZZ,h^0\parens{\booltype}}$.

\begin{figure*}
  \centering
  \begin{tikzpicture}
    \matrix[row sep=13mm,column sep=10mm] {
      \node (hX1X0) {$h^1\parens{X_1/X_0}$};
      &
      \node (hX0) {$h^0\parens{X_0}$};
      &
      \node (h2-hX0) {$h^0\parens{\booltype} \times h^0\parens{X_0}$};
      &
      \node (h2) {$h^0\parens{\booltype}$};
      \\
      \node (A1h2) {$\displaystyle \prod_{A_1} h^0\parens{\booltype}$};
      &
      \node (A0neh2) {$\displaystyle \prod_{\sum_{x:A_0} (a \ne x)} h^0\parens{\booltype}$};
      &
      \node (h2-A0neh2) {$\displaystyle h^0\parens{\booltype} \times \prod_{\sum_{x:A_0} (a \ne x)} h^0\parens{\booltype}$};
      &
      \\
      &
      &
      \node (A0h2) {$\displaystyle \prod_{A_0} h^0\parens{\booltype}$};
      \\
      \node (ZA1-h2) {$\displaystyle \hom\parens{\formalsum{A_1},h^0\parens{\booltype}}$};
      &
      &
      \node (ZA0-h2) [xshift=-7mm] {$\displaystyle \hom\parens{\formalsum{A_0},h^0\parens{\booltype}}$};
      &
      \node (Z-h2) [xshift=7mm] {$\displaystyle \hom\parens{\ZZ,h^0\parens{\booltype}}$};
      \\
    };
    \draw [arrow] (Z-h2) -- (ZA0-h2) node [edgelabel, below] {$\hom\parens[\big]{\augmentation,h^0\parens{\booltype}}$};
    \draw [arrow] (ZA0-h2) -- (ZA1-h2) node [edgelabel, below] {$\hom\parens[\big]{\boundary_1,h^0\parens{\booltype}}$};
    \draw [equiv] (h2|-Z-h2.north) -- (h2) node [edgelabel, right] {$k_{-1}$};
    \draw [equiv] (A0h2|-ZA0-h2.north) -- (A0h2);
    \draw [equiv] (ZA1-h2) -- (A1h2);
    \draw [equiv] (A0h2) -- (h2-A0neh2) node [edgelabel, right] {\parbox[t]{3cm}{tricky; \\ see text}};
    \draw [arrow] (h2-A0neh2) -- (A0neh2) node [edgelabel, above] {$\snd$};
    \draw [arrow] (A0neh2) -- (A1h2);
    \draw [equiv] (h2-A0neh2) -- (h2-hX0);
    \draw [equiv] (A0neh2) -- (hX0);
    \draw [equiv] (A1h2) -- (hX1X0);
    \draw [arrow] (h2) -- (h2-hX0) node [edgelabel, above] {$\coboundary_0$};
    \draw [arrow] (h2-hX0) -- (hX0) node [edgelabel, above] {$\snd$};
    \draw [arrow] (hX0) -- (hX1X0) node [edgelabel, above] {Lemma~\ref{lma:LES}};
  \end{tikzpicture}
  \caption{A commutative diagram showing how the commutative square for $n = -1$ was broken down into smaller pieces. The commutativity of each smaller square is a critical step in the proof.}
  \label{fig:commute:degree-zero}
\end{figure*}

The zeroth dimension needs special attention
for the same reason given in the proof of Lemma~\ref{lma:zero-as-wedge}:
the point $a$ is used as the center of the wedge.
Therefore, we have to add one copy of $h^0\parens{\booltype}$ to the left hand side
to make ends meet.
A more detailed breakdown is the commutative diagram shown in Figure~\ref{fig:commute:degree-zero}.
There is one more subtlety regarding the middle isomorphism of the third column:
it would be naive to directly collect values arising from
$h^0\parens{\booltype}$ and $h^0\parens{X_0}$ (through $\prod_{\sum_{x:A_0}a\ne x} h^0\parens{\booltype}$)
and form a function in the product group $\prod_{A_0} h^0\parens{\booltype}$.
That unfortunately would not make all the needed squares commute.
The correct way to merge an element $g$ from $h^0\parens{\booltype}$
with an element $f$ from $\prod_{\sum_{x:A_0}a\ne x} h^0\parens{\booltype}$ is to construct a function $\bar{f}$ such that
\[
  \bar{f}(x) =
  \begin{cases}
  g &\text{~if $a=x$}
  \\
  f(x) \cdot g &\text{~if $a\ne x$}.
\end{cases}
\]
where the binary operation is the one
specified by the group $h^0\parens{\booltype}$.
The intuition is that in cohomology
we care about the difference between values on points,
not their absolute values;
thus, the natural way to maintain the difference
is to calculate the relative values
when the center $a$ is chomped
and restore them when the center $a$ revives.
The first component from $h^0\parens{\booltype}$ in the isomorphism
may be seen as the standard sea level
for the second component.

Once the correct isomorphisms are in place,
the commutative squares are relatively straightforward.
The simplest proof we found is to rephrase
the $\coboundary_n$ and the dualized $\boundary_n$,
through the isomorphisms $k_n$
and Proposition~\ref{prop:h-is-hom},
as homomorphisms between the groups
$\prod_{A_n} h^0\parens{\booltype}$
and
$\prod_{A_{n+1}} h^0\parens{\booltype}$.
When the sets $A_n$ are finite,
all these maps can be rephrased as the following:
\[
  \lam{f}\lam{x}[A_{n+1}]\sum_{y:A_{n}}{f(y)}^{\deg(\sigma_{x,y})}
\]
for some endo-arrows $\sigma_{x,y}$ of the spheres.
(Note that the $\sum$ here is the summation defined by the group structure,
not the sigma types.)
Therefore, it boils down to proving that the parameters $\sigma$ for the $\coboundary_n$
and the dualized $\boundary_n$ have the same degrees.
In our setting, the parameters for coboundary maps
are always the suspensions of those for boundary maps,
and thus share the same degree by the property~\ref{it:deg-susp}
of the degree function. This concludes the proof of Lemma~\ref{lma:cochain}.
The significance of the finiteness of $A_n$ is that
it turns group products into direct sums
and consequently gives an explicit description
of the inverse map in the \axiom{wedge} axiom of the \EilenbergSteenrod\ framework
and enables summation over cells at a dimension.

\section{Mechanization in Agda}

The main theorem (Theorem~\ref{thm:main}) along with the two lemmas (Lemmas~\ref{lma:reformulate}~and~\ref{lma:cochain})
have been mechanized in the proof assistant Agda,
using the \texttt{HoTT-Agda} library~\cite{hott-in:agda}.
The statement of the main theorem can be found at
\burl{https://github.com/HoTT/HoTT-Agda/blob/f46225f94285205724/theorems/cw/cohomology/AxiomaticIsoCellular.agda#L18-L25}.

In order to mechanize the proofs in our work,
significant improvements have been done to the general framework
and it is difficult to quantify them.
Already there are more than 4,000 lines of code
written specifically for the work presented in this paper,
among which much engineering was done to make it type-check
within reasonable time and memory space.

The code has been organized into two major parts:
\begin{description}
  \item[Restruction of cochain complex]
    This corresponds to Lemma~\ref{lma:reformulate}.
    Note that this part only assumes
    a limited form of the axiom of choice (used by the \axiom{wedge} axiom),
    decidable equality on $A_n$ (for boundary maps),
    and that degree maps have finite supports.
    Finiteness of $A_n$ is not explicitly required.
  \item[Equivalence between two cochain complexes]
    This corresponds to Lemma~\ref{lma:cochain}.
\end{description}

\section{Conclusion and further work}
\label{sec:conclusion}

The main theorem in this paper dealt with the case where every cell set $A_n$
was finite and the dimension of the complex is finite.
A first, relatively straight-forward extension,
is to consider infinite-dimensional
cell complexes with finite cell sets.

We believe that it should be possible to further generalize
the result to (certain) infinite cell sets as follows.
Recall that the spheres are \emph{$\omega$-compact}~\cite[Chapter 7]{egbert-as:thesis}: if
$X : \omega \to \UU$ is a sequential diagram of types, then any
$f : \sphere n \to \varinjlim_k X_k$ factors through some $X_k$.
More precisely:
\begin{defi}
  A type $C$ is \emph{$\omega$-compact} if for any sequential diagram
  $X : \omega \to \UU$, the canonical map
  \[
    \varinjlim\nolimits_k (C \to X_k)
    \to \bigl(C \to \varinjlim\nolimits_k X_k\bigr)
  \]
  is an equivalence.
\end{defi}
\begin{conj}
  If each $A_n$ is either finite or isomorphic to $\NN$, then each stage of a cell complex $X_n$
  is a sequential colimit.
\end{conj}
Depending on how precisely this is formalized,
it is possible that countable choice is needed:
there should be a natural map from
a sequential colimit of finite subcomplexes into $X_n$,
and $\omega$-compactness of spheres should ensure
this is an equivalence.

Note that classically,
almost all cell complexes we want to consider have
explicitly given cell sets that are either finite
or isomorphic to $\NN$.
It is also worth remarking that assuming
countable choice may not interfere with
a computational interpretation of type theory,
since countable choice holds
in realizability models of extensional type theory.
(On the other hand,
countable choice is not provable in homotopy type theory,
as shown in~\cite{coquand-mannaa-ruch}.)
We leave these deliberations for future work.

Another obvious line of inquiry
is to use our setup to do cellular homology theory.
Once a good library has been developed for homology theory,
we in fact expect that the result corresponding
to our Theorem~\ref{thm:main} will be simpler
in the homology setting.

\smallskip

We can obtain a direct comparison between our results
concerning bare CW complexes,
and results about usual CW complexes as topological spaces
using the real-cohesive type theory of~\cite{shulman-brouwer}.
A CW complex in topological spaces
corresponds to a CW complex in the sense of Section~\ref{sec:cw},
but with all spheres replaced by their topological counterparts,
and the top map in the pushout defining $X_{n+1}$ replaced by
$\mathrm{id}_{A_{n+1}} \times \mathrm{in}_n$,
where $\mathrm{in}_n$ is the inclusion of the unit sphere
into the unit ball in $\mathbb{R}^{n+1}$.
The shape functor will take such a topological CW complex
to a CW complex in our sense,
since it preserves colimits and maps the topological spheres and balls
to their homotopy types.

\smallskip

Eventually, it is also our hope that our work can be part of
a full formalization of practical algorithms
used in constructive algebraic topology,
such as those in the Kenzo system~\cite{kenzo}.

\section*{Acknowledgment}

\newcommand{\grantsponsor}[3]{#2}
\newcommand{\grantnum}[2]{#2}

  We want to acknowledge the valuable comments
  that we received
  on earlier drafts and presentations
  of the work in this paper
  from
  Carlo Angiuli, 
  Guillaume Brunerie, 
  Evan Cavallo, 
  Floris van Doorn, 
  Robert Harper, 
  Egbert Rijke, 
  and anonymous reviewers. 

  This research was sponsored by
  the \grantsponsor{nsf}{National Science Foundation}{https://www.nsf.gov/}
  under grant number \grantnum{nsf}{DMS-1638352}
  and
  the \grantsponsor{afosr}{Air Force Office of Scientific Research}{http://www.wpafb.af.mil/afrl/afosr/}
  under grant number \grantnum{afosr}{FA9550-15-1-0053}.
  The authors would also like to thank the Isaac Newton Institute for Mathematical Sciences
  for its support and hospitality during the program ``Big Proof''
  when part of the work on this paper was undertaken;
  the program was supported by \grantsponsor{esprc}{Engineering and Physical Sciences Research Council}{https://www.epsrc.ac.uk/}
  under grant number \grantnum{esprc}{EP/K032208/1}.
  The views and conclusions contained in this document are those of the authors
and should not be interpreted as representing the official policies, either expressed or implied, of any
sponsoring institution, government or any other entity.

\bibliographystyle{alphaurl}
\bibliography{cohomology}

\newcommand{\etalchar}[1]{$^{#1}$}
\begin{thebibliography}{HFFLL16}

\bibitem[AH16]{chtt.part2}
Carlo Angiuli and Robert Harper.
\newblock Computational higher type theory {II}: Dependent cubical
  realizability, 2016.
\newblock \href {http://arxiv.org/abs/1606.09638} {\path{arXiv:1606.09638}}.

\bibitem[AHH17]{chtt.part3}
Carlo Angiuli, Kuen-Bang {Hou (Favonia)}, and Robert Harper.
\newblock Computational higher type theory {III}: Univalent universes and exact
  equality, 2017.
\newblock \href {http://arxiv.org/abs/1712.01800} {\path{arXiv:1712.01800}}.

\bibitem[AHH18]{ccctt-at:csl}
Carlo Angiuli, Kuen-Bang {Hou (Favonia)}, and Robert Harper.
\newblock {Cartesian Cubical Computational Type Theory: Constructive Reasoning
  with Paths and Equalities}.
\newblock In {\em 27th EACSL Annual Conference on Computer Science Logic (CSL
  2018)}, volume 119 of {\em Leibniz International Proceedings in Informatics
  (LIPIcs)}, pages 6:1--6:17, Dagstuhl, Germany, 2018. Schloss
  Dagstuhl--Leibniz-Zentrum fuer Informatik.
\newblock \href {https://doi.org/10.4230/LIPIcs.CSL.2018.6}
  {\path{doi:10.4230/LIPIcs.CSL.2018.6}}.

\bibitem[AHW16]{chtt.part1}
Carlo Angiuli, Robert Harper, and Todd Wilson.
\newblock Computational higher type theory {I}: Abstract cubical realizability,
  2016.
\newblock \href {http://arxiv.org/abs/1604.08873} {\path{arXiv:1604.08873}}.

\bibitem[AHW17]{chtt}
Carlo Angiuli, Robert Harper, and Todd Wilson.
\newblock Computational higher-dimensional type theory.
\newblock In {\em Proceedings of the 44th ACM SIGPLAN Symposium on Principles
  of Programming Languages (POPL 2017)}, pages 680--693, New York, NY, USA,
  2017. ACM.
\newblock \href {https://doi.org/10.1145/3009837.3009861}
  {\path{doi:10.1145/3009837.3009861}}.

\bibitem[Ark95]{arkowitz}
Martin Arkowitz.
\newblock Co-{$H$}-spaces.
\newblock In {\em Handbook of algebraic topology}, pages 1143--1173.
  North-Holland, Amsterdam, 1995.
\newblock \href {https://doi.org/10.1016/B978-044481779-2/50024-9}
  {\path{doi:10.1016/B978-044481779-2/50024-9}}.

\bibitem[BGL{\etalchar{+}}17]{hott-in:coq}
Andrej Bauer, Jason Gross, Peter~LeFanu Lumsdaine, Michael Shulman, Matthieu
  Sozeau, and Bas Spitters.
\newblock The {HoTT} library: A formalization of homotopy type theory in {C}oq.
\newblock In {\em Proceedings of the 6th ACM SIGPLAN Conference on Certified
  Programs and Proofs}, CPP 2017, pages 164--172, New York, NY, USA, 2017. ACM.
\newblock \href {https://doi.org/10.1145/3018610.3018615}
  {\path{doi:10.1145/3018610.3018615}}.

\bibitem[BH89]{berstein-harper}
Israel Berstein and John~R. Harper.
\newblock Cogroups which are not suspensions.
\newblock In Gunnar Carlsson, Ralph Cohen, Haynes Miller, and Douglas Ravenel,
  editors, {\em Algebraic topology, Proc. Int. Conf., Arcata/Calif. 1986},
  volume 1370 of {\em LNM}, pages 63--86, Heidelberg, 1989. Springer.
\newblock \href {https://doi.org/10.1007/BFb0085219}
  {\path{doi:10.1007/BFb0085219}}.

\bibitem[BH18]{cohomology-in:hott}
Ulrik Buchholtz and Kuen-Bang {Hou (Favonia)}.
\newblock Cellular cohomology in homotopy type theory.
\newblock In {\em Proceedings of the 33rd Annual ACM/IEEE Symposium on Logic in
  Computer Science}, LICS '18, pages 521--529, New York, NY, USA, 2018. ACM.
\newblock \href {https://doi.org/10.1145/3209108.3209188}
  {\path{doi:10.1145/3209108.3209188}}.

\bibitem[BHC{\etalchar{+}}18]{hott-in:agda}
Guillaume Brunerie, Kuen-Bang {Hou (Favonia)}, Evan Cavallo, Eric Finster,
  Jesper Cockx, Christian Sattler, Chris Jeris, Michael Shulman, et~al.
\newblock Homotopy type theory in {A}gda, 2011--2018.
\newblock URL: \url{https://github.com/HoTT/HoTT-Agda}.

\bibitem[BR17]{realprojective}
Ulrik Buchholtz and Egbert Rijke.
\newblock The real projective spaces in homotopy type theory.
\newblock In {\em 32nd Annual ACM/IEEE Symposium on Logic in Computer Science
  (LICS 2017)}, pages 1--8. IEEE, New York, NY, USA, 2017.
\newblock \href {https://doi.org/10.1109/LICS.2017.8005146}
  {\path{doi:10.1109/LICS.2017.8005146}}.

\bibitem[Buc17]{cellular.complexes-in:utt-by:ulrik}
Ulrik Buchholtz.
\newblock Cellular complexes in {Lean}, 2017.
\newblock URL:
  \url{https://github.com/leanprover/lean2/blob/master/hott/homotopy/cellcomplex.hlean}.

\bibitem[Cav15]{mayer.vietoris-in:utt}
Evan Cavallo.
\newblock Synthetic cohomology in homotopy type theory.
\newblock Master's thesis, 2015.
\newblock URL: \url{http://www.cs.cmu.edu/~ecavallo/works/thesis.pdf}.

\bibitem[CCHM18]{cchm}
Cyril Cohen, Thierry Coquand, Simon Huber, and Anders M{\"o}rtberg.
\newblock {Cubical Type Theory: A Constructive Interpretation of the Univalence
  Axiom}.
\newblock In Tarmo Uustalu, editor, {\em 21st International Conference on Types
  for Proofs and Programs (TYPES 2015)}, volume~69 of {\em Leibniz
  International Proceedings in Informatics (LIPIcs)}, pages 5:1--5:34,
  Dagstuhl, Germany, 2018. Schloss Dagstuhl--Leibniz-Zentrum fuer Informatik.
\newblock \href {https://doi.org/10.4230/LIPIcs.TYPES.2015.5}
  {\path{doi:10.4230/LIPIcs.TYPES.2015.5}}.

\bibitem[CH18]{chtt.part4}
Evan Cavallo and Robert Harper.
\newblock Computational higher type theory iv: Inductive types, 2018.
\newblock \href {http://arxiv.org/abs/1801.01568} {\path{arXiv:1801.01568}}.

\bibitem[CHM18]{cubicaltt-hits}
Thierry Coquand, Simon Huber, and Anders M\"{o}rtberg.
\newblock On higher inductive types in cubical type theory.
\newblock In {\em Proceedings of the 33rd Annual ACM/IEEE Symposium on Logic in
  Computer Science}, LICS '18, pages 255--264, New York, NY, USA, 2018. ACM.
\newblock \href {https://doi.org/10.1145/3209108.3209197}
  {\path{doi:10.1145/3209108.3209197}}.

\bibitem[CMR17]{coquand-mannaa-ruch}
Thierry Coquand, Bassel Mannaa, and Fabian Ruch.
\newblock Stack semantics of type theory.
\newblock In {\em 32nd Annual ACM/IEEE Symposium on Logic in Computer Science
  (LICS 2017)}, pages 1--11, New York, NY, USA, 2017. IEEE.
\newblock \href {https://doi.org/10.1109/LICS.2017.8005130}
  {\path{doi:10.1109/LICS.2017.8005130}}.

\bibitem[Doo16]{prop-trunc}
Floris~van Doorn.
\newblock Constructing the propositional truncation using non-recursive hits.
\newblock In {\em Proceedings of the 5th ACM SIGPLAN Conference on Certified
  Programs and Proofs (CPP 2016)}, pages 122--129, New York, NY, USA, 2016.
  ACM.
\newblock \href {https://doi.org/10.1145/2854065.2854076}
  {\path{doi:10.1145/2854065.2854076}}.

\bibitem[DvRB17]{hott-in:lean}
Floris~van Doorn, Jakob von Raumer, and Ulrik Buchholtz.
\newblock Homotopy type theory in {Lean}.
\newblock In {\em Interactive Theorem Proving (ITP 2017)}, pages 479--495,
  Cham, 2017. Springer.
\newblock \href {https://doi.org/10.1007/978-3-319-66107-0_30}
  {\path{doi:10.1007/978-3-319-66107-0_30}}.

\bibitem[ES45]{eilenberg.steenrod.cohomology.axioms}
Samuel Eilenberg and Norman~E. Steenrod.
\newblock Axiomatic approach to homology theory.
\newblock {\em Proceedings of the National Academy of Sciences of the United
  States of America}, 31(4):117--120, 1945.
\newblock URL: \url{http://www.pnas.org/content/31/4/117.full.pdf}.

\bibitem[Gra17]{homology-in-hott}
Robert Graham.
\newblock Synthetic homology in homotopy type theory, 2017.
\newblock \href {http://arxiv.org/abs/1706.01540} {\path{arXiv:1706.01540}}.

\bibitem[GSS98]{kenzo}
Julio~Rubio Garcia, Francis Sergeraert, and Yvon Siret.
\newblock Kenzo: A symbolic software for effective homology computation, 1998.
\newblock URL: \url{https://www-fourier.ujf-grenoble.fr/~sergerar/Kenzo/}.

\bibitem[Hat02]{hatcher-at}
Allen Hatcher.
\newblock {\em Algebraic Topology}.
\newblock Cambridge University Press, Cambridge, UK, 2002.
\newblock URL: \url{https://www.math.cornell.edu/~hatcher/AT/ATpage.html}.

\bibitem[HF17]{favonia-as:thesis}
Kuen-Bang Hou~(Favonia).
\newblock {\em Higher-Dimensional Types in the Mechanization of Homotopy
  Theory}.
\newblock PhD thesis, Carnegie Mellon University, 2017.
\newblock URL: \url{http://favonia.org/thesis}.

\bibitem[HFFLL16]{blakers-massey-in-hott}
Kuen-Bang Hou~(Favonia), Eric Finster, Daniel~R. Licata, and Peter~LeFanu
  Lumsdaine.
\newblock A mechanization of the blakers-massey connectivity theorem in
  homotopy type theory.
\newblock In {\em Proceedings of the 31st Annual IEEE Symposium on Logic in
  Computer Science (LICS 2016)}, pages 565--574, New York, NY, USA, 2016. IEEE.
\newblock \href {https://doi.org/10.1145/2933575.2934545}
  {\path{doi:10.1145/2933575.2934545}}.

\bibitem[LB13]{licata-brunerie-pinsn}
Daniel~R. Licata and Guillaume Brunerie.
\newblock {$\pi_n(\mathbb{S}^n)$} in homotopy type theory.
\newblock In Georges Gonthier and Michael Norrish, editors, {\em Proceedings of
  the 3rd International Conference on Certified Programs and Proofs (CPP
  2013)}, pages 1--16, Cham, 2013. Springer International Publishing.
\newblock \href {https://doi.org/10.1007/978-3-319-03545-1_1}
  {\path{doi:10.1007/978-3-319-03545-1_1}}.

\bibitem[May]{may-concise}
J.~Peter May.
\newblock A concise course in algebraic topology.
\newblock URL:
  \url{https://www.math.uchicago.edu/~may/CONCISE/ConciseRevised.pdf}.

\bibitem[Rij17]{joinconstruction}
Egbert Rijke.
\newblock The join construction, 2017.
\newblock Preprint.
\newblock \href {http://arxiv.org/abs/1701.07538} {\path{arXiv:1701.07538}}.

\bibitem[Rij18]{egbert-as:thesis}
Egbert Rijke.
\newblock {\em Classifying Types}.
\newblock PhD thesis, Carnegie Mellon University, 2018.
\newblock \href {http://arxiv.org/abs/1906.09435} {\path{arXiv:1906.09435}}.

\bibitem[Shu13]{cohomology-in:utt-by:mike}
Michael Shulman.
\newblock Cohomology, 2013.
\newblock URL: \url{http://homotopytypetheory.org/2013/07/24/cohomology/}.

\bibitem[Shu17]{shulman-brouwer}
Michael Shulman.
\newblock Brouwer's fixed-point theorem in real-cohesive homotopy type theory.
\newblock {\em Mathematical Structures in Computer Science}, pages 1--86, 2017.
\newblock \href {https://doi.org/10.1017/S0960129517000147}
  {\path{doi:10.1017/S0960129517000147}}.

\bibitem[{Uni}13]{hott-as:book}
The {Univalent Foundations Program}.
\newblock {\em Homotopy Type Theory: Univalent Foundations of Mathematics}.
\newblock Institute for Advanced Study, 2013.
\newblock URL: \url{http://homotopytypetheory.org/book}, \href
  {http://arxiv.org/abs/1308.0729} {\path{arXiv:1308.0729}}.

\end{thebibliography}
\end{document}